\documentclass[12pt]{article}
\usepackage[utf8]{inputenc}
\usepackage{lmodern}
\usepackage[T1]{fontenc}
\usepackage{mathptmx}
\usepackage{eufrak}

\usepackage{graphicx}
\usepackage{amsthm,amsmath,amssymb}

\usepackage{enumitem}
\usepackage{setspace}
\usepackage{nicefrac}
\usepackage{natbib}
\usepackage{bibentry}

\usepackage{caption}
\usepackage{bm}

\usepackage{tikz-cd}
\usetikzlibrary{positioning,calc}
\usetikzlibrary{arrows, arrows.meta, positioning, automata, decorations.pathreplacing, decorations.markings, patterns}
\usepackage{float}
\usepackage{forest}
\usepackage{mathtools}
\usepackage{hyperref}
\hypersetup{
	colorlinks   = true, %Colours links instead of ugly boxes
	urlcolor     = blue, %Colour for external hyperlinks
	linkcolor    = blue, %Colour of internal links
	citecolor   = blue %Colour of citations
}

\usepackage[defaultcolor=red]{changes} % this package helps in tracking changes
\usepackage[a4paper,margin=3.2cm]{geometry}
\usepackage{titlesec}
\titleformat*{\section}{\large\bfseries}
\titleformat*{\subsection}{\bfseries}
\newtheorem{mydef}{Definition}
\newtheorem{lemma}{Lemma}
\newtheorem{proposition}[lemma]{Proposition}
\newtheorem{theorem}[lemma]{Theorem}
\newtheorem{propnum}{Proposition} % used in appendix for second counter
\newtheorem{lemmanum}[propnum]{Lemma}   % shares the counter with propnum

\makeatletter
\newtheorem*{rep@theorem}{\rep@title}
\newcommand{\newreptheorem}[2]{%
	\newenvironment{rep#1}[1]{%
		\def\rep@title{#2 \ref{##1}}%
		\begin{rep@theorem}%
		}{%
		\end{rep@theorem}%
	}%
}
\makeatother
\newreptheorem{theorem}{Theorem}

\begin{document}

	\def\thefootnote{\fnsymbol{footnote}}
	
	\begin{center}
		{\Large\bf Choice via Dual Ownership}\footnote{We thank Franz Dietrich, Federico Echenique, Sean Horan, Ali Khan, Marco Mariotti, Alessandro Pavan, and John Quah for valuable comments. Financial support under PRIN 20222Z3CR7 ``Nudging under limited attention'' is acknowledged. The authors have no relevant financial or non-financial interests to disclose. They applied Borda ownership to their stylistic choices, and used the AI tool Claude for copyediting and to check proofs. The authors reviewed and edited the content as needed and take full responsibility for it.}
	\end{center}
	
	\setcounter{footnote}{0}
	\def\thefootnote{\alph{footnote}}
	
	\vspace*{.1cm}
	\begin{center}
		\begin{tabular}{ccc}
			{\sc Bhavook Bhardwaj} & $\quad$ & {\sc Marco LiCalzi} \\
			%{\small Dept.\ of Economics} & & {\small Venice School of Management}\\
			{\small Université de Montréal} & & {\small Università Ca' Foscari Venezia} \\
			{\small\tt bhavook.bhardwaj@umontreal.ca} & & {\small\tt licalzi@unive.it}\\
		\end{tabular}
	\end{center}
	\vspace*{0.2cm}
	\begin{center}
		(November 2026)
	\end{center}
	\vspace*{0.4cm}

	{\small
		\noindent {\bf Abstract.} 
		When there are multiple agents (or rationales), an ownership rule assigns the choice for each menu to one of them. We study choice under dual ownership, where decision control alternates between two preferences. The main contribution is the behavioral characterization of a rule that endogenously decides ownership to maximize consensus. We also axiomatize two simpler exogenous rules. Identification results are provided for each rule.
		
		
	}
	
	\setcounter{footnote}{0}
	\def\thefootnote{\arabic{footnote}}
	
	\bigskip
	
	\section{\sc \textbf{Introduction}}\label{sec:introduction}
	
	Rational choice theory formalizes the purposeful and consistent pursuit of one's goals as the maximization of a single preference relation. A seminal contribution by \citet{kalai2002rationalizing} suggests that departures from classical rationality may be explained by the coexistence of multiple binary relations, called \emph{rationales}.
	
	While the existing literature explores numerous choice procedures based on multiple binary relations (see Section~\ref{sec:dual}), it has often neglected to explicitly consider how the competition between conflicting rationales is resolved. We address this gap by introducing the concept of \emph{ownership}.
	
	An ownership rule maps each menu to the agent who makes the choice. Classical rationality corresponds to constant ownership, where a single preference dictates decisions across all menus. In our framework, the agent who ``owns'' a decision need not be an individual. There may be two selves (e.g., present vs.\ future) competing over a consumption/saving decision, or two rationales vying to drive a choice. Therefore, we simply say that an ownership rule assigns a preference to each menu, and leave it unspecified to whom or what the preference refers. The following three examples illustrate different ways in which an ownership rule may apply. 
	
	\smallskip
	\noindent
	\textbf{Example 1.} A principal delegates some decisions to a subordinate but retains authority when the scale or the stakes increase; see \citet{aghion1997formal}. Similarly, a parent lets a child decide on small menus and keeps the larger ones for herself. 
	
	\smallskip
	\noindent
	\textbf{Example 2.} An individual alternates between rapid and contemplative selves. Simple binary choices are handled by the rapid self, whereas larger menus trigger the contemplative self.
	
	\smallskip
	\noindent
	\textbf{Example 3.} A pair of experts have individual rankings over the options in a menu. If their top choices coincide, this is the collective choice. If they conflict, control is assigned to the expert whose top choice is ranked relatively higher by the other expert.
	
	\smallskip
	We focus on \emph{dual ownership} rules, where there are exactly two preferences.\footnote{Ownership with many preferences is hard to tackle. \citet{apesteguia2010computational} and \citet{demuynck2011computational} show that the problem of rationalizing choice behavior with multiple rationales is NP-complete.} This is analogous to the approach of \citet{manzini2018dual}, whose \emph{Dual Random Utility Maximisation} (dRUM) studies a stochastic choice function generated by a menu-dependent probability distribution over two preferences. The dRUM model assigns ownership conditional on random (unspecified) events. 
	
	Returning to non-stochastic choice, \citet{eliaz2011choosing} consider procedures in which the choice is a subset of (at most) two alternatives. In particular, their \emph{top-and-top} (TAT) procedure selects the alternatives that are top ranked by at least one preference.\footnote{Appendix~\ref{app:TAT} offers an alternative axiomatization for TAT.} One may say that the choice is co-owned, because each preference contributes equally to the final choice. 
	
	We study the observable implications of ownership for (non-stochastic) choice theory. Section~\ref{sec:setup} sets up the model and formalizes the concept of a dual ownership rule. Section~\ref{sec:dual} reviews the literature and proves that sequentially rationalizable choice \citep{manzini2007sequentially} is not nested in dual ownership. Section~\ref{sec:menu} gives behavioral characterizations for two menu-dependent ownership rules, formalizing Examples~1 and~2. Section~\ref{sec:borda} axiomatizes choice behavior under the endogenous ownership rule from Example~3; we view this as our main technical contribution. Section~\ref{sec:ident} discusses identification for all three models. Section~\ref{sec:end} concludes.
	
	\section{ \sc \textbf{The model}}\label{sec:setup}
	
	Let $X$ be a finite set with cardinality at least~$2$. A \emph{menu} is a non-empty subset of $X$. Let $\mathcal{P}(X)$ denote the collection of menus from $X$. Following \citet{cherepanov2013rationalization}, a \emph{preference} is an asymmetric binary relation $\succ$ on $X$; and a connected\footnote{A binary relation $\succ$ is connected if, for any $x \ne y$, either $y \succ x$ or $x \succ y$. This is similar to completeness, but it does not impose $x \succ x$.}, transitive preference is an \emph{order}. Given a preference $\succ$, the \emph{maximal set} $\{x \in A : \neg (y \succ x) \ \text{for all} \ y \in A\}$ for a menu $A$ may be empty and need not be a singleton. When the maximal set is a singleton, we write $\max(A,\succ)$ for its unique element. A \emph{choice correspondence} $c$ on $\mathcal{P}(X)$ assigns a non-empty\footnote{Trivially, $c( \{ x \}) = \{x\}$; for simplicity, we ignore single-element menus in the following.} subset $c(A) \subseteq A$ to every menu $A$. If $c$ is single-valued for every menu $A$, we call it a \emph{choice function} and identify the singleton $c(A)$ with its unique element.
	
	Imagine a pair of agents, endowed with preferences $\succ_1$ and $\succ_2$. Intuitively, the novel idea in this paper is that, for each menu, the pair designates one of the two as a proxy and lets the proxy choose. The agents can be individuals (e.g., parent and child), selves (e.g., fast and slow modes of thinking), or rationales (e.g., egoistic and altruistic rankings): what matters is that the choice for each menu is imputed to an agent $i=1,2$ who ``owns'' the decision and has a preference that picks $\max \left(A, \succ_i \right)$.
	
	Formally, consider a pair of preferences $O= \{\succ_1, \succ_2\}$ on $X$. A (menu-dependent) \emph{ownership rule} $\omega : \mathcal{P}(X) \to O$ associates a preference in $O$ with each menu in $\mathcal{P}(X)$. The choice behavior generated via an ownership rule can be viewed as a two-step process: first, the rule $\omega$ selects which preference to deploy; second, the choice maximizes that preference.
	
	\begin{mydef}\label{def1} 
		A choice function $c$ on $\mathcal{P}(X)$ respects \emph{single ownership} if there exists a pair of preferences $O= \{\succ_1, \succ_2\}$ on $X$ and an ownership rule $\omega : \mathcal{P}(X) \to O$ such that, for any $A \in \mathcal{P}(X)$,
		\[
		c(A) = \max \left(A, \omega (A) \right).
		\]
	\end{mydef}
	
	Figure~\ref{fig1-cvo} unpacks Definition~\ref{def1} by juxtaposing the observables of a choice function and the model of choice via single ownership. For any menu $A$, the model delivers a choice $c(A)$ by identifying the owner of the decision and by imposing that this owner acts as a maximizer.
	\begin{figure}[ht!]
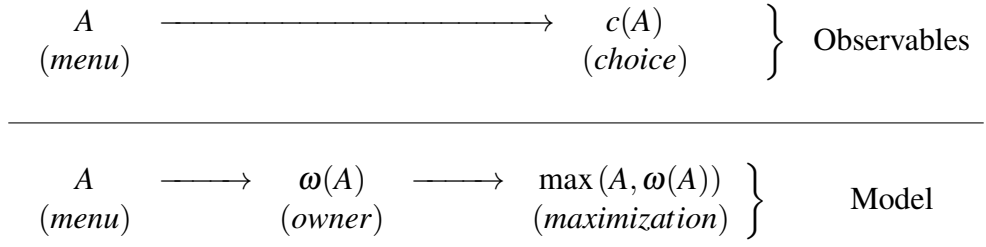

		\[
		\begin{array}{lcr}
			% Column 1: The nested array with its own brace
			\left.
			\begin{array}{ccccc}
				A & \multicolumn{3}{c}{\xrightarrow{\hspace{5cm}}} & c(A) \\
				(menu) & & & & (choice) \\
			\end{array}
			\hspace{.8cm} \right\} & \text{Observables} \\[8mm] \hline \\
			\left.
			\begin{array}{ccccc}
				A & \xrightarrow{\hspace{1cm}} & \omega(A) & \xrightarrow{\hspace{1cm}} & \max \left(A, \omega (A) \right)  \\
				(menu) & & (owner) & & (maximization)
			\end{array}
			\right\} 
			& \text{Model}
		\end{array}
		\]
		\caption{Choice via ownership}\label{fig1-cvo}
	\end{figure}
	
	Moving further, one can consider joint ownership, where each of the two agents may designate a best choice, as in the TAT procedure \citep{eliaz2011choosing}. To formalize joint ownership, we extend the (single) ownership in Definition~\ref{def1}. Let $ \mathcal{O} = \{\{ \succ_1\}, \{\succ_2\}, \{\succ_1, \succ_2\}\}$ be the collection of possible \emph{ownership arrangements}. 
	
	\begin{mydef}\label{def2} 
		A choice correspondence $c$ on $\mathcal{P}(X)$ respects \emph{joint ownership} if there exists a pair of preferences $O= \{\succ_1, \succ_2\}$ on $X$ and an ownership arrangement $\omega : \mathcal{P}(X) \rightarrow \mathcal{O}$ such that, for any $A \in \mathcal{P}(X)$,
		\[
		c(A) = \bigcup_{\succ_i \in \omega(A)} \max (A, \succ_i).
		\]
		
	\end{mydef}
	
	An example of a joint ownership rule is the following. Agent~1 is an engineer with expertise in a domain $D_1 \subset X$ and Agent~2 is a lawyer with expertise in a disjoint domain $D_2 \subset X$. For any menu $A$ within the domain of expertise $D_i$, the decision is owned uniquely by Agent~$i$ and $c(A) $ maximizes $\succ_i$; otherwise, the decision is jointly owned and $c(A)$ delivers the maximizers for $\succ_1$ and $\succ_2$. We provide an alternative axiomatization for TAT in Appendix~\ref{app:TAT}.
	
	Note that $c(\cdot)$ is a function for single ownership and a correspondence for joint ownership; this is necessary because the latter may deliver distinct choices for the two owners. Differences in observability may also explain this distinction, since the modeler might see the two top choices without knowing who owns them, or only the final choice.\footnote{The case where the modeler observes both the owners and their two top choices is equivalent to rationalizing two standard preferences.} A similar issue applies to the dRUM, which has a dual support but delivers a unique realization: we can interpret it as a model where ties between top choices are broken randomly or where we randomly observe one of the two top choices.
	
	We have defined single and joint ownership as menu-dependent rules. More generally, however, ownership may depend on context, framing, or relational properties; see \citet{dietrich2016reason} for a thoughtful discussion. In particular, Example~3 in the introduction describes a rule that also depends on the two preferences restricted to the menu. Our last definition covers the case in which the ownership arrangement $\omega (A, \succ_1, \succ_2)$ depends both on the menu $A$ and on the preferences $\succ_1$ and $\succ_2$.
	
	\begin{mydef}
		A choice correspondence $c$ respects \emph{general ownership} if there exists a pair of preferences $O= \{\succ_1, \succ_2\}$ on $X$ and an ownership arrangement $\omega : \mathcal{P}(X) \times O^2 \to \mathcal{O}$ such that, for any $A \in \mathcal{P}(X)$,
		\[
		c(A) = \bigcup_{\succ_i \in \omega(A,\succ_1, \succ_2)} \max \left(A, \succ_i \right).
		\]
		
	\end{mydef}
	
	In all three definitions the range of the ownership rule consists of (at most) two preferences; we refer to this property as \emph{dual ownership}.
	
	\section{\sc\textbf{Dual ownership and choice theory}}\label{sec:dual}
	
	This section reviews relevant literature in light of the notion of ownership, and demonstrates how this notion can be used to categorize a large number of different models. For brevity, we focus on deterministic choice; however, there is a wealth of models of stochastic choice that corroborate our categories. This section can be skipped with no significant loss of continuity.
	
	\citet{kalai2002rationalizing} exemplify their model of rationalization by multiple rationales (RMR) with three examples. We draw on them to illustrate the main distinctions. The first example is borrowed from \citet[p.~288]{Luce1957}: a customer at a restaurant switches his order from $a$ to $b$ upon learning that a third never-chosen option (frog's legs) is available. The ranking is revised because a restaurant that offers frog's legs is in a different league. We say that ownership is \emph{(strictly) menu-dependent} when the ranking used to make the choice depends only on the menu.\footnote{A minimal generalization of the frog's legs example is to assume that the ownership rule preserves $i$'s \emph{primacy} if there is a subset $A^*$ such that $A \cap A^* \ne \emptyset$ implies $\omega(A) = \succ_i$. Intuitively, choice is entrusted to $\succ_i$ whenever an alternative from $A^*$ is available.}
	
	We study two simple models of deterministic choice under menu-dependent ownership in Section~\ref{sec:menu}. They respectively formalize Examples~1 and~2 from the introduction. We are not aware of other models from the literature in the same (deterministic) class. The dRUM in \citet{manzini2018dual}, on the other hand, deals with stochastic choice under menu-dependent ownership. 
	
	Clearly, menu-dependent ownership is only a special case. Indeed, the second example in \citet{kalai2002rationalizing} demonstrates that changes in ownership need not depend only on the menu. Their $(u,v)$ procedure assumes two orders and a threshold $t \in X$: given a menu $A$, the choice is $x=\max(A,\succ_1)$ if $x \succ_2 t$; otherwise, it is $\max(A,\succ_2)$. Ownership is transferred from $\succ_1$ to $\succ_2$ when $\succ_1$ cannot find an alternative that is ``satisficing'' for $\succ_2$. In this case, we say that ownership is \emph{reference-dependent}. \citet{apesteguia2009theory} study a model of deterministic choice in which ownership depends on an exogenously given reference point. \citet{ok2015revealed} and \citet{kibris2023theory} derive the reference points endogenously.
	
	In general, many different factors may affect which ranking is used to make a choice \citep{dietrich2016reason}. For instance, ownership may depend on frames \citep{salant2008frames} or endowments \citep{masatlioglu2014canonical}. 
	
	In Section~\ref{sec:borda}, we study a third model of deterministic choice in which ownership over a menu is endogenously assigned to maximize consensus, as measured by the non-owner's ranking. \citet{cuhadaroglu2017choosing} considers a different model of dual ownership in which the two agents exert influence on each other. 
	
	The third and last example in \citet{kalai2002rationalizing}, revisited in~\citet[p.~2189]{masatlioglu2012revealed}, assumes that the agent restricts attention to the most popular items in the menu before making a choice. There is no change in ownership, because the final choice always calls on the same ranking; however, when presented with a menu $A$, the owner ends up choosing from a smaller set $B \subseteq A$. We say that ownership is \emph{restricted}.
	
	The literature on models with restricted ownership is rich. The prominent class of sequentially rationalizable choice\footnote{These include \citet{au2011sequentially}, \citet{masatlioglu2012revealed}, \citet{manzini2012categorize}, \citet{apesteguia2013choice}, \citet{horan2016simple}.} originates from the rational shortlist method (RSM) in~\citet{manzini2007sequentially}. This is based on the sequential application of rationales: initially, to narrow the menu; subsequently, to make a choice over the constricted menu. The ownership over the final choice never changes, but it is exercised over a smaller set than the original menu. Subsection~\ref{sec:rodual} at the end of this section shows that the restricted ownership underlying RSM is not nested in dual ownership. 
	
	More generally, restricted ownership is a consequence of different limitations that filter the menu processed by the owner. The owner may consciously rule out alternatives that cannot be justified \citep{cherepanov2013rationalization}, or be unable to take into consideration all the feasible options \citep{masatlioglu2012revealed, dean2017attention, lleras2017more}. A related strand of the literature explores situations in which a primary self restricts ownership by deciding the menu over which the second self is allowed to make a choice \citep{kreps1979representation,gul2001temptation}. 
	
	Finally, there are procedures that input the owners' rankings and deliver a choice that may not be at the top of any ranking. We say that choice is \emph{unowned}. A typical example is \citet{ambrus2015rationalising}, where individual rankings aggregate into a collective ranking used to make the final choice. \citet{de2012reason} study a model of dual ownership in which choice is a compromise between the owners' rankings, using fallback bargaining. Our model in Section~\ref{sec:borda} is technically similar, but makes the opposite assumption that the final choice has to be optimal for at least one agent.
	
	\subsection{\sc \textbf{Restricted ownership is not dual ownership}}\label{sec:rodual}
	
	One may conjecture that unrestricted and restricted ownership should lead to distinct classes of models. This claim must be qualified. \citet[Proposition~1]{kalai2002rationalizing} prove that RMR can rationalize choice over any finite set when the number of rationales is sufficiently high. Therefore, the proper (qualitative) conjecture should state that unrestricted ownership with few owners is generally incompatible with restricted ownership. 
	
	We corroborate this conjecture with an example showing that the choice behavior generated by the rational shortlist method (RSM) has no dual ownership representation. We rely on the following lemma.
	
	\begin{lemma}\label{lem:nested}
		Let $c$ be a choice function. Given two menus $A \subset B$ with $c(B) \in A$ and $c(A) \neq c(B)$, no single preference $\succ$ satisfies both $c(A) = \max(A,\succ)$ and $c(B) = \max(B,\succ)$. Therefore, in any ownership representation, $A$ and $B$ have different owners.
	\end{lemma}
	\begin{proof}
		Suppose by contradiction that $\succ$ rationalizes both choices. Since $c(B)$ is $\succ$-maximal in $B$ and $c(B) \in A \subset B$, it is also $\succ$-maximal in $A$, yielding the contradiction $c(A) = \max(A,\succ)= c(B)$.
	\end{proof}
	
	Let $X = \{x_1, x_2, x_3, d, e\}$. The first preference is $\succ_1 = \{(d,x_1),(d,x_2),(e,x_1)\}$, so that $d \, \succ_1 \, x_1$, $d \, \succ_1 \, x_2$ and $e \, \succ_1 \, x_1$; the second is the order $\succ_2$ given by $x_1 \succ_2 x_2 \succ_2 x_3 \succ_2 d \succ_2 e$. Let $c$ denote the choice generated by RSM: from each menu, eliminate the alternatives that are $\succ_1$-dominated within the menu, then use $\succ_2$ to choose among the survivors. Consider
	\[
	A = \{x_1,x_2,x_3\}, \qquad B = A \cup \{e\}, \qquad C = A \cup \{d,e\}.
	\]
	Then $A \subset B \subset C$, with:
	\begin{center}
		\begin{tabular}{lll}
			menu & $\succ_1$-survivors & $c$ \\ \hline
			$A$ & $x_1,x_2,x_3$ & $x_1$ \\
			$B$ & $x_2,x_3,e$ & $x_2$ \\
			$C$ & $x_3,d,e$ & $x_3$ \\
		\end{tabular}
	\end{center}
	The three choices are distinct and all belong to $A$. Applying Lemma~\ref{lem:nested} to the pairs $(A,B)$, $(A,C)$ and $(B,C)$ shows that $A$, $B$ and $C$ must have three different owners, whereas dual ownership provides at most two. Hence, $c$ has no dual ownership representation.
	
	\section{\sc\textbf{Menu-dependent ownership}}\label{sec:menu}
	
	A simple example of menu dependence arises when ownership is assigned according to the size of the menu: one agent handles the small problems, and the other deals with the large ones.\footnote{When set inclusion is ``an indicator of \emph{relative complexity}'' \cite[p.~54]{tyson2008cognitive}, one can replace ``small'' and ``large'' with ``simple'' and ``difficult''.} 
	
	This section considers two models of menu-dependent ownership that generalize this example in two different directions, corresponding respectively to Examples~1 and~2 in the introduction. We provide revealed preference characterizations. Identification results are in Section~\ref{sec:ident}. 
	
	\subsection{\sc \textbf{Monotone ownership}}
	
	Ownership is \emph{monotone} with respect to set inclusion if, once a menu grows large enough that ownership switches to the other agent, it can no longer return to the first one. See the \emph{Weak WARP} property in \citet{manzini2007sequentially} for a related notion. 
	
	On the left, Figure~\ref{fig3} illustrates a (color-coded) monotone ownership rule: blue for $\succ_1$, red for $\succ_2$. Monotonicity implies that a cut partitions the lattice $\mathcal{P}(X)$ into two ownership domains.
	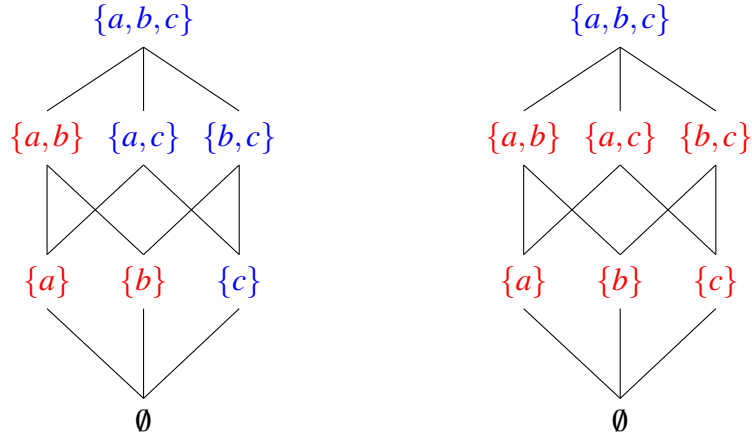
\begin{figure}[htbp]
		\begin{center}
			\begin{tikzpicture}
				\matrix (A) [matrix of nodes, row sep=1.2cm]
				{
					$\color{blue}\{a,b\}$ & $\color{red}\{a,c\}$ & $\color{red}\{b,c\}$ \\
					$\color{blue}\{a\}$ & $\color{blue}\{b\}$ & $\color{red}\{c\}$ \\
				};
				\path (A-1-1)--(A-1-2) node[above=1.2cm] (link) {$\color{red}\{a,b,c\}$};
				
				\foreach \i in {1,...,3}
				\draw (link.south) -- (A-1-\i.north);
				\foreach \i/\j in {1/2, 3/2, 2/1, 1/1, 3/3, 2/3}
				\draw (A-1-\i.south)--(A-2-\j.north);
			\end{tikzpicture}
			\hfil
			\begin{tikzpicture}
				\matrix (A) [matrix of nodes, row sep=1.2cm]
				{
					$\color{blue}\{a,b\}$ & $\color{blue}\{a,c\}$ & $\color{blue}\{b,c\}$ \\
					$\color{blue}\{a\}$ & $\color{blue}\{b\}$ & $\color{blue}\{c\}$ \\
				};
				\path (A-1-1)--(A-1-2) node[above=1.2cm] (link) {$\color{red}\{a,b,c\}$};
				
				\foreach \i in {1,...,3}
				\draw (link.south) -- (A-1-\i.north);
				\foreach \i/\j in {1/2, 3/2, 2/1, 1/1, 3/3, 2/3}
				\draw (A-1-\i.south)--(A-2-\j.north);
			\end{tikzpicture}
			\caption{A monotone ownership rule (left) and a size-monotone rule (right).}\label{fig3}
		\end{center}
	\end{figure}
	The stronger notion of \emph{size monotonicity} allows for switches in ownership only when the menu reaches a given size. As shown on the right in Figure~\ref{fig3}, this levels the cut to a horizontal boundary. 
	
	Monotonicity fits situations like Example~1, where one agent is in charge of small menus and the other of large menus. When the size of a menu increases the complexity of making a decision, one can say that one agent handles the simple choices and the other the difficult ones. 
	
	We model dual monotone ownership assuming that $c(\cdot)$ is a choice function and the two binary relations $\succ_1, \succ_2$ are preferences (i.e., asymmetric). This weaker assumption features prominently in \citet{manzini2007sequentially}. 
	
	\begin{mydef}
		An ownership rule is \emph{monotone} if, for any three menus $A \subset B \subset C$, $\omega (A) \neq \omega(B)$ implies $\omega(B)=\omega(C)$.
	\end{mydef}
	
	A monotone ownership rule partitions the set of menus into two blocks, respectively called \emph{upper set} and \emph{downset}. An upper set $U \subseteq \mathcal{P}(X)$ satisfies upward closure: if $A$ is in $U$, then any superset $B \supset A$ is in $U$. Similarly, a downset satisfies downward closure. Intuitively, the downset and the upper set match the antonyms small/large or simple/difficult. 
	
	Our goal is to characterize the choice functions generated by a monotone ownership rule. It is well known that contraction consistency (or Sen's condition $\alpha$) is a necessary condition for the rationalizability of a choice function. We propose a relaxation of condition $\alpha$.
	
	\medskip
	\noindent
	\emph{\textbf{Weak-$\alpha$}}: For all menus $B \subset B^\prime \subseteq A \subset A^\prime$, if $c(A^\prime)\in A$ and $c(A) \ne c(A^\prime)$, then $c(B^\prime ) \in B$ implies $c(B)=c(B^\prime )$.
	\medskip
	
	This weaker condition requires that, if contracting $A^\prime$ to $A$ violates consistency, then further contractions of $A$ do not violate consistency. Put simply, a violation of consistency may occur at most once along a chain of menus ordered by set inclusion. Because a violation reveals an ownership change, monotonicity requires no further ownership flips along the same chain. Our first result shows that Weak-$\alpha$ characterizes monotone ownership. (All omitted proofs for Section~\ref{sec:menu} are in Appendix~\ref{app:menu}.) 
	
	\begin{theorem}\label{thm1}
		A choice function respects monotone ownership if and only if it satisfies Weak-$\alpha$.
	\end{theorem}
	
	Here is the main intuition. Weak-$\alpha$ partitions the set of menus into two blocks, respectively closed under union and under intersection. The standard condition $\alpha$ holds separately for each block. But a choice function that satisfies condition $\alpha$ on a domain that is closed under unions (or intersections) has a rationalizing preference over that domain. 
	
	Note that the preference over the domain closed under unions is acyclic, and thus can be extended to an order. This fits well with the parent-child paradigm, where one expects the first to exhibit more rationality than the second: the parent's choices over larger menus are rationalized by an order, whereas the child's choices over smaller menus are rationalized by a preference. 
	
	\subsection{\sc \textbf{Simple ownership}}\label{sec:simple}
	
	Example~2 suggests a slightly more general rule that assigns ownership over all binary menus to the same agent and is otherwise unrestricted. We still assume that $c(\cdot)$ is a choice function, but require the two binary relations $\succ_1, \succ_2$ to be orders. Formally, we say that the ownership rule is \emph{simple} if $\omega(A)= \omega(B)$ for all menus $A, B$ of size~2.
	
	We characterize simple ownership using two axioms. The first one requires that choice over binary menus rules out cycles of length~3. We omit the obvious quantifiers for readability.
	
	\medskip
	\noindent
	\emph{\textbf{Pairwise consistency (PC)}}: If $c(\{x,y\})=x$ and $c(\{y,z\})=y$, then $c(\{x,z\})=x$.
	\medskip
	
	The second axiom relaxes the \emph{Strong Axiom of Revealed Preference (SARP)}: the revealed preference need not be acyclic over the full domain but only across specific menus of size at least~3. Before stating it, we need the following definition.
	
	\begin{mydef}
		An alternative $x$ \emph{prevails} over another alternative $y$, and we write $x \vartriangleright y$, if there exists a sequence of alternatives $x_1,\dots,x_{k+1}$ with $x_1=x$ and $x_{k+1}=y$, and a sequence of menus $A_1,\dots,A_k$, such that for each $i \in \{1,\dots,k\}$:
		\begin{enumerate}[itemsep=1pt,topsep=3pt]
			\item $x_i,x_{i+1}\in A_i$ and $c(A_i)=x_i$; but
			\item $c(\{x_i,z\})=z\,$ for some $z \in A_i$, with $z \neq x_i$.
		\end{enumerate}
	\end{mydef}
	
	Condition 2 of the definition reveals that the choice in each $A_i$ contradicts the order used for binary menus. Therefore, all menus $A_i$ must be owned by a second preference. We require prevalence to be asymmetric to ensure that this second preference is acyclic, and thus can be extended to an order.
	
	\medskip
	\noindent
	\emph{\textbf{Asymmetric prevalence (AP)}}:
	If $x \vartriangleright y$, then $\neg \left( y \vartriangleright x\right)$.
	\medskip
	
	Our next result characterizes choice behavior under a simple ownership rule. The independence of the two axioms is shown in Appendix~\ref{app:indep}.
	
	\begin{theorem}\label{thm2}
		A choice function respects simple ownership if and only if it satisfies PC and AP.
	\end{theorem}
	
	The special case of Example~2, where menus are partitioned into simple (size at most 2) and difficult (size at least 3), can be characterized by modifying the definition of prevalence: assume that each menu $A_i$ has size at least 3 and drop Condition~2. 
	
	\section{\sc\textbf{Borda ownership}}\label{sec:borda}
	
	Given two orders and a menu $A$, the fallback bargaining solution axiomatized in~\citet{de2012reason} picks the alternative(s) that maximize the worst rank\footnote{Given an order $ \succ$, the \emph{rank} of an alternative $x$ in a menu $A$ is the number of alternatives in $A$ that $\succ$ ranks below $x$. This is called \emph{score} in~\citet{de2012reason} or Borda count in social choice.} attributed by either order within $A$ itself. (In case of a tie, it picks both.) The choice depends on the menu and on the orders. This procedure is unowned, because it can deliver choices that no agent ranks first. However, it suggests that ownership rules may depend not only on the menu but also on how two rankings relate to each other. 
	
	Our third model formalizes Example~3 through one such endogenous ownership rule. Observe that the rank of an alternative in $A$ under $\succ_i$ is its Borda count, denoted $b_i (x, A)$; for instance, if a menu $A$ has $n$ elements, the rank (and Borda count) of the top alternative $x$ for $\succ_1$ is $b_1 (x, A) = n-1$, while $b_2(x,A)$ may be any integer up to $n-1$. 
	
	The \emph{Borda ownership} rule picks the top alternative that has the higher sum of Borda counts over the two orders. Equivalently, because at least one of the two Borda counts for a top alternative in $A$ is $|A|-1$, it selects the top alternative ranked higher by the other order. (In case of a tie, it picks both.) Formally, for any menu $A$, denote by $x_i^*$ its top alternative according to $\succ_i$; the Borda ownership rule is 
	\[
	\omega^B (A, \succ_1, \succ_2) = \begin{cases} \{\succ_1\} & \text{if } b_{2}(x_1^*, A) > b_1(x_2^*, A) \\
		\{\succ_2\} & \text{if } b_{2}(x_1^*, A) < b_1(x_2^*, A) \\
		\{\succ_1, \succ_2\} & \text{if } b_{2}(x_1^*, A) = b_1(x_2^*, A)
	\end{cases}
	\]	
	Intuitively, given that the choice has to be a top alternative, Borda ownership is assigned so as to maximize the rank attributed by the non-owner to the owner's choice. Note that this rule assigns a single owner, unless the ranks attributed by each agent to the other's top choice are the same. 
	
	Compare Borda ownership with the fallback bargaining solution: both describe procedures to make a choice when there are two agents with potentially different orders. In the spirit of bargaining, the fallback solution seeks a compromise that takes everyone into account. Borda ownership, instead, insists on attributing choice to an agent, and it assigns ownership to minimize the non-owner's discontent.
	
	We characterize choice behavior via Borda ownership assuming that $c(\cdot)$ is a choice correspondence and the two binary relations $\succ_1, \succ_2$ are orders. Some concepts are inspired by similar notions in~\citet[Section~5]{de2012reason}: our characterization shares three of the eight axioms used in their Theorem~2. However, the motivation and the bulk of the argument are quite different.
	
	We begin with two notions based on choices over pairs and triplets. An alternative $x$ is revealed to (pairwise) \emph{dominate} $y$, denoted $x \gg y$, if $c(\{x,y\})=\{x\}$. Under Borda ownership, $x \gg y$ if and only if $x \succ_i y$ for both $i$, so pairwise dominance may be read as Pareto dominance. We also say that $x$ and $y$ are \emph{incomparable}, denoted $x \| y$, if neither dominates the other.
	
	An alternative $y$ is revealed to be \emph{intermediate} between $x$ and $z$, denoted $x \langle y \rangle z$, if the choice from each pair in the triplet\footnote{For brevity, we call \emph{triplet} any 3-element set without assuming that its elements are ordered.} is the pair itself and $c(\{x,y,z\})= \{x,z\}$. Intuitively, $y$ is intermediate when the two orders rank $x$ and $z$ in opposite ways, and both place $y$ between them. Dominance coincides with Pareto superiority in~\citet{de2012reason}, while our intermediate differs from their compromise, because the latter uses $c(\{x,y,z\})= \{y\}$. 
	
	Our first three axioms are borrowed from~\citet{de2012reason}. The next three are variants of their axioms adapted to our setting. The last axiom is new. The first axiom is pairwise consistency (PC) from the previous section, restated as transitivity of the dominance relation.
	
	\medskip
	\noindent
	\emph{\textbf{Pairwise Consistency (PC)}}: If $x \gg y$ and $y \gg z$, then $x \gg z$.
	\medskip
	
	The second axiom requires that expanding a menu with a dominated alternative $y$ preserves exactly those original choices that dominate $y$. \citet{de2012reason} introduce it to model the attraction effect \citep{huber1982adding}.
	
	\medskip
	\noindent	
	\emph{\textbf{Attraction (ATT)}}: Let $x \in c(A)$ and $y \notin A$ be such that $x \gg y$. Then $c(A \cup \{y\}) = c(A) \cap \left\{z \in A \mid z \gg y \right\}$.
	\medskip
	
	The third axiom states that no dominated alternative can be chosen; that is, the choice correspondence is a selection from the set of non-dominated alternatives.
	
	\medskip
	\noindent	
	\emph{\textbf{Efficiency (EFF)}}: If $x \in c(A)$ then there exists no $y \in A$ such that $y \gg x$.
	\medskip
	
	The next axiom states that, when the choice from every pair of alternatives in a triplet is non-resolute, the choice from the triplet itself is a doubleton, revealing that the third alternative is an intermediate. This is a variant of the axiom named Existence of a Compromise in~\citet{de2012reason}: where they ask for a singleton, we require a doubleton. Their singleton is directly revealed as the compromise; our doubleton indirectly reveals the intermediate. 
	
	Intuitively, they seek to promote the choice of a compromise, while we want to avoid the choice of an intermediate. This intuition applies to each of the next three axioms.
	
	\medskip
	\noindent	
	\emph{\textbf{Existence of an Intermediate (EI)}}: If the choice out of every pair in $\{ x,y,z \}$ is the pair itself, then $c(\{x,y,z\})$ is a doubleton.
	\medskip
	
	If $y$ is an intermediate between $x$ and $z$, one order puts $x$ above $y$ and the other order puts $z$ above $y$; thus $y$ cannot be a top alternative in a menu where $x,z$ are present. The fifth axiom states that if $y$ is intermediate between $x$ and $z$, denoted $x \langle y \rangle z$, then $y$ cannot be chosen when the menu contains $x,z$. This is a variant of the No Better Compromise axiom. 
	
	\medskip
	\noindent		
	\emph{\textbf{No Better Intermediate (NBI)}}: If $x,z \in c(A)$ and $x \langle y \rangle z$, then $y \notin c(A)$.
	\medskip
	
	The sixth axiom	 concerns any intermediate $y$ between $x$ and $z$. If a fourth alternative $w$ is incomparable to $y$, then $y$ must also be an intermediate between $w$ and one of $x,z$. To illustrate, suppose $x\langle y \rangle z$ because $x \succ_1 y \succ_1 z$ and $z \succ_2 y \succ_2 x$. If neither $w \gg y$ nor $y \gg w$, then (say)\ $w \succ_1 y$ but $y \succ_2 w$: we require $w \succ_1 y \succ_1 z$ and $z \succ_2 y \succ_2 w$. This is a variant of the Overcoming a Compromise axiom.
	
	\medskip
	\noindent
	\emph{\textbf{Staying Intermediate (SI)}}: If $x\langle y \rangle z$ and $c(\{y,w\}) = \{y,w\}$ for some $w \notin \{x,y,z\}$, then $x \langle y \rangle w$ or $z \langle y \rangle w$.
	\medskip
	
	The last one is an expansion axiom: it concerns adding alternatives to a menu. It deals with two cases: (i) adding two dominated alternatives, or (ii) adding an intermediate. The first part requires that adding to a menu $A$ two more alternatives that are respectively dominated by distinct alternatives in $c(A)$ does not change the chosen items. The second part imposes that adding to a menu an intermediate between alternatives in $c(A)$ does not change the chosen items. In both cases, the addition does not change the top alternatives and leaves the non-owner's Borda counts for the two chosen alternatives equal; hence, under Borda ownership, the original choice cannot change.
	
	\medskip
	\noindent	
	\emph{\textbf{Expansion (EXP)}}:
	\begin{enumerate}[itemsep=0pt, topsep=0pt, parsep=0pt]
		\item[(i)] if $x,y \in c(A)$ and $z, w \notin A$ are such that $x \gg z$ but $x \| w$ whereas $y \gg w$ but $y \| z$, then $x,y \in c(A \cup \{z, w\})$;
		\item[(ii)] if $x,z \in c(A)$ and $y \notin A$ are such that $x \langle y \rangle z$, then $x,z \in c(A \cup \{y\})$.
	\end{enumerate}
	\medskip
	
	Our next result characterizes choice behavior under a Borda ownership rule. The proof is in Appendix~\ref{app:borda}. The independence of the axioms is shown in Appendix~\ref{app:indep}.
	
	\begin{theorem}\label{thm3}
		A choice correspondence respects Borda ownership if and only if it satisfies PC, ATT, EFF, EI, NBI, SI and EXP.
	\end{theorem}
	
	\section{\sc \textbf{Identification}}\label{sec:ident}
	
	Sections~\ref{sec:menu} and~\ref{sec:borda} characterize choice behavior for three ownership rules. This section deals with their identification from observed choices. For monotone or simple ownership, a \emph{representation} of $c$ consists of a pair $(\succ_1,\succ_2)$ and a menu-dependent ownership rule $\omega$. For Borda ownership, it consists only of two orders, because ownership is endogenously computed from them. 
	
	As is typical with multiple-rationale models, we should not expect a unique representation. Our results identify the restrictions common to all representations. Relabeling the two agents never changes $c$, so we do not count it as a different representation. All proofs are in Appendix~\ref{app:ident}.
	
	\subsection{\sc \textbf{Monotone ownership}}
	
	For monotone ownership, identification hinges on violations of Weak-$\alpha$. Given a choice function $c$, a pair of menus $(A,B)$ is a \emph{reversal} if $A \subset B$, $c(B) \in A$ and $c(A) \neq c(B)$. If the same agent owned both menus, $c(B)$ would be the best element of $B$ and hence of $A$, so we would have $c(A)=c(B)$. Therefore, the two menus of a reversal have different owners, with monotonicity resolving who owns what around the reversal. 
	
	\begin{proposition}\label{prop:id1} 
		Let $c(\cdot)$ be a choice function that satisfies Weak-$\alpha$. 
		\begin{enumerate}[label=(\arabic*),itemsep=1pt,topsep=3pt]
			\item A monotone ownership rule $\omega$ is part of a representation of $c$ if and only if it assigns different owners to the two menus of every reversal.
			\item If Agent~$i$ owns a menu $A$ in a representation, then no $y \in A$ satisfies $y \succ_i' c(A)$. In particular, if Agent~$i$ owns the menu $\{x,y\}$, then $x \succ_i y$ if and only if $c(\{x,y\})=x$. 
		\end{enumerate}
	\end{proposition}
	
	If $c$ exhibits a reversal and Agent~1 owns its smaller menu, then in every representation Agent~1 owns all submenus of the smaller menu, and Agent~2 owns all supersets of the larger menu. Intuitively, a reversal reveals a change of ownership, and ownership monotonicity tells us who owns the menus around the reversal. Part~(2) says that the owner's choice is unbeaten, so his choice over an owned pair reveals his preference.
	
	\subsection{\sc \textbf{Simple ownership}}
	
	Under simple ownership one agent, say Agent~1, owns every pair, so $\succ_1$ is read off directly from binary choices: $x \succ_1 y$ if and only if $c(\{x,y\})=x$. Then $\succ_1$ is an order by PC. Let $\mathcal{D}$ be the collection of menus $A$ with $|A|\ge 3$ whose choice contradicts this order; that is, such that $c(A) \neq \max(A,\succ_1)$. The menus in $\mathcal{D}$ cannot belong to Agent~1, so they must be owned by Agent~2. The prevalence relation $x \vartriangleright y$ of Section~\ref{sec:simple} is the transitive closure of the pairs $(c(A),y)$ with $A \in \mathcal{D}$ and $y \in A \setminus \{c(A)\}$. By AP, this is a strict partial order.
	
	\begin{proposition}\label{prop:id2}
		Let $c(\cdot)$ be a choice function that satisfies PC and AP. A triple $(\succ_1', \succ_2', \omega)$ represents $c$ under simple ownership (with Agent~1 owning pairs) if and only if
		\begin{enumerate}[label=(\arabic*),itemsep=0pt]
			\item $\succ_1' = \succ_1$;
			\item $\succ_2'$ is an order extending $\vartriangleright$ (i.e., $x \vartriangleright y$ implies $x \succ_2' y$);
			\item Agent~2 owns every menu in $\mathcal{D}$, and owns any other menu $A$ with $|A|\ge3$ only if $c(A)=\max(A,\succ_2')$.
		\end{enumerate}
	\end{proposition}
	
	In short, Agent~1 is fully identified, while Agent~2 is identified exactly as far as prevalence goes: any order extending $\vartriangleright$ will do. Part~(3) is the analog for ownership: a menu whose choice is consistent with the pairwise order can still be assigned to Agent~2, unless some alternative in the menu prevails over the choice.
	
	\subsection{\sc \textbf{Borda ownership}}
	
	Borda ownership is endogenously computed from the two orders, so identification concerns the orders only. Call a non-empty set $B \subseteq A$ \emph{dominant in $A$} if $x \gg y$ for all $x \in B$ and $y \in A \setminus B$. Two dominant sets are nested (otherwise $x \gg y \gg x$). Let $B_1$ be the smallest dominant set in $X$, let $B_2$ be the smallest dominant set in $X \setminus B_1$, and so on. This yields a partition $B_1,\dots,B_K$ of $X$ into \emph{blocks}. Every alternative in $B_k$ dominates every alternative in $B_\ell$ for $k<\ell$, so the two orders rank the blocks in the same way and each block is an interval for both orders. Thus, agents can disagree only inside a block. For a set $J$ of blocks, let $\succ^J$ be the pair of orders obtained from $(\succ_1,\succ_2)$ by exchanging the two agents' orders inside each block of $J$.
	
	\begin{proposition}\label{prop:id3}
		Let $c(\cdot)$ be generated by the orders $(\succ_1,\succ_2)$ under Borda ownership. 
		\begin{enumerate}[label=(\arabic*),itemsep=1pt,topsep=3pt]
			\item A profile $\succ^\prime$ generates $c$ if and only if $\succ^\prime = \succ^J$ for some set of blocks $J$. 
			\item Let $m$ be the number of blocks with at least two alternatives. Up to relabeling, there are $2^{m-1}$ representations if $m \ge 1$ and a single one if $m=0$; hence, identification is unique if and only if $m \le 1$. 
		\end{enumerate}
	\end{proposition}
	
	The intuition is that only the first block met by a menu matters for its choice, and that a block can be flipped as a whole because Borda ownership treats the two agents symmetrically.
	
	\section{\sc \textbf{Conclusions}}\label{sec:end}
	
	The notion of an ownership rule disciplines how agents, selves, or rationales compete to drive a choice. We have operationalized three ownership rules, providing revealed preference characterizations of choice behavior for each of them. 
	
	A few research directions emerge. First, ownership rules can be studied in structured settings such as household decisions \citep{browning1998efficient} or choice under uncertainty. In structurally richer environments, ownership may depend on many elements, including contextual cues. Second, one may follow up with the study of random ownership and extend the analysis to stochastic choice, aiming to provide novel explanations for boundedly rational choice behavior. Third, choice via ownership rules can be based on more than two binary relations. This makes general characterizations considerably less tractable, but we believe that plausible restrictions on the ownership rules along with more structure in the choice environments may be a promising avenue to explore.
	
	% ======================================================================
	% References
	% ======================================================================
	
	\setlength{\bibsep}{0.1 cm} 
	\bibliographystyle{ecta} 
	\bibliography{bibliog.bib} 	
	
	% ======================================================================
	% Appendix
	% ======================================================================
	
	\newpage
	\appendix
	\section{\sc \textbf{Appendix}}
	
	For brevity, when a quantifier involves two elements $x,y$, they are understood to be distinct. In many cases, allowing for $x=y$ is inconsequential. We also write $xy$ for the menu $\{x,y\}$, $xyz$ for $\{x,y,z\}$, and so on.
	
	\subsection{\sc \textbf{Proofs for Section~\ref{sec:menu}}}\label{app:menu}
	
	\begin{reptheorem}{thm1}
		A choice function respects monotone ownership if and only if it satisfies Weak-$\alpha$.
	\end{reptheorem}
	\begin{proof} \textit{(Necessity)} Assume that the choice function $c$ respects a (single) monotone ownership rule $\omega$. Let $c(A) \ne c(A^\prime) \in A \subset A^\prime$ and, without loss of generality, let $\omega (A^\prime) = \succ_1$. Because $c (\cdot)$ is single-valued, we have $ c(A^\prime) = \max(A^\prime,\succ_1)$. But $c(A) \ne \max \left( A^\prime,\succ_1 \right)$ for $A \subset A^\prime$ implies $\omega(A)= \succ_2$. Consider any $B \subset B^\prime \subseteq A$. Then $\omega (B) = \omega(B^\prime )= \succ_2$ by monotonicity of $\omega$, and thus if $c(B^\prime )=\max(B^\prime, \succ_2)$ belongs to $B$, then $c(B^\prime ) = \max (B,\succ_2) = c(B)$. 
		
		\medskip\noindent
		\textit{(Sufficiency)} Consider the collection of sets for which removing one suitable element violates condition $\alpha$: let 
		\[
		\mathcal{C}^* = \{A \in \mathcal{P}(X) \ | \text{ for some } x \notin A, \ c(A) \neq c(A\cup \{x\}) \text{ but } c(A \cup \{x\}) \in A \}.
		\] 
		We distinguish two cases: first, $\mathcal{C}^*=\emptyset$; second, $\mathcal{C}^* \neq \emptyset$. 
		
		\smallskip\noindent
		\textit{Case 1:} $\mathcal{C}^*=\emptyset$. We first claim that $c$ satisfies condition $\alpha$ if and only if $\mathcal{C}^*=\emptyset$. 
		
		Necessity is immediate: if $c$ satisfies condition $\alpha$, then $c(A\cup \{x\})\in A$ implies $c(A)=c(A \cup \{x\})$ and thus $\mathcal{C}^*=\emptyset$. For sufficiency, let $\mathcal{C}^* = \emptyset$ and assume by contradiction that $c$ violates condition $\alpha$. Then there exist two menus $A \subset A^\prime$ such that $c(A^\prime ) \in A$ and $c(A) \neq c(A^\prime )$. Let $c(A)=x$ and $c(A^\prime )=y$. Consider any sequence of sets $A = A_0 \subset A_1 \subset \ldots \subset A_k=A^\prime$ such that $x,y \in A_i$ for all $i \in \{0,\ldots,k\}$ and $|A_{i+1}| = |A_i| + 1$ for all $i \in \{0,\ldots, k-1\}$. Because $c(A_0)=x$ and $c(A_k)=y$ with $y\in A_i$ for all $i$, we must have some $j \in \{0,\ldots ,k-1\}$ such that $c(A_{j+1}) \in A_{j}$ and $c(A_j)\neq c(A_{j+1})$. Then $A_j \in \mathcal{C}^*$, contradicting $\mathcal{C}^*= \emptyset$. This establishes the first claim.
		
		This claim implies that, if $\mathcal{C}^*= \emptyset$, condition $\alpha$ holds. Hence, $c$ is a (standard) rational choice function with full domain: we set $\mathcal{C}_1 = \mathcal{P}(X)$, $\mathcal{C}_2 = \emptyset$ and $\omega(A)= \succ_1$ for any menu $A$, where $\succ_1$ is the rational choice ordering derived from $c$ (and $\succ_2 = \succ_1$).
		
		\smallskip\noindent
		\textit{Case 2:} $\mathcal{C}^*\neq \emptyset$. Partition the collection of menus into the two sets $\mathcal{C}_1 = \{A \in \mathcal{P}(X) \ | \ A \subseteq A^\prime \text{ for some } A^\prime \in \mathcal{C}^* \}$ and $\mathcal{C}_2 = \mathcal{P}(X)\setminus\mathcal{C}_1$. Note that $\mathcal{C}_2$ is closed under unions and that $X \in \mathcal{C}_2$. 
		
		We now claim that the restriction of $c$ to $\mathcal{C}_2$ satisfies condition $\alpha$. Suppose not. Then there exist menus $A, A^\prime \in \mathcal{C}_2$ with $A \subset A^\prime$, $c(A^\prime) \in A$ and $c(A)\neq c(A^\prime )$. However, observe that since $A \in \mathcal{C}_2$, we must have $A^{\prime\prime} \in \mathcal{C}_2$ for any $A^{\prime\prime} \supset A$. Because $\mathcal{C}_2$ is closed under unions, similarly to Case~1, we can find an increasing sequence $A = A_0 \subset A_1 \subset \ldots \subset A_k=A^\prime$ in $\mathcal{C}_2$ with $c(A_{i+1}) \in A_i$ and $c(A_i) \neq c(A_{i+1})$ for some $i$. This would imply $A_i \in \mathcal{C}^* \subseteq \mathcal{C}_1$, which is a contradiction. This establishes the second claim.
		
		Let $x \succ_2 y$ hold if there exists $A \in \mathcal{C}_2$ with $x = c(A)$ and $y \in A$. Then condition $\alpha$ ensures that $\succ_2$ is an acyclic preference that rationalizes $c$ over $\mathcal{C}_2$; see \citet[Chapter~3]{rubinstein2012lecture}. For any $A \in \mathcal{C}_2$, we define $\omega(A) = \succ_2$ and assign ownership to Agent~2.
		
		We now consider $\mathcal{C}_1$, which is closed under intersections.\footnote{Let $\partial\mathcal{C}^* = \{A \in \mathcal{C}^* \ | \ A \subset A^\prime \text{ implies } A^\prime \notin \mathcal{C}^* \}$ be the upper contour of $\mathcal{C}_1$. Then $\mathcal{C}_1$ is the collection of subsets of members of $\partial\mathcal{C}^*$.} Let $x \succ_1 y$ hold if there exists $A \in \mathcal{C}_1$ with $x =c(A)$ and $y \in A$. We show that the binary relation $\succ_1$ is asymmetric, and hence it is a preference. Assume by contradiction that there exist elements $x,y$ and menus $A,A^\prime \in \mathcal{C}_1$ such that (i) $x = c(A)$ and $y \in A$; and (ii) $c(A^\prime )=y$ and $x \in A^\prime $. Choose $M, M^\prime \in \mathcal{C}^*$ with $A \subseteq M$ and $A^\prime \subseteq M^\prime$. For $M \in \mathcal{C}^*$ there is $z \notin M$ with $c(M \cup \{z\}) \in M$ and $c(M) \neq c(M \cup \{z\})$; then Weak-$\alpha$ implies no reversals across subsets of $M$, and likewise for $M^\prime$. But $x = c(A) \in A \cap A^\prime \subseteq M$ implies $c(A \cap A^\prime) = x$ and $y = c(A^\prime) \in A \cap A^\prime \subseteq M^\prime$ implies $c(A \cap A^\prime) = y$, which is a contradiction. The preference $\succ_1$ rationalizes all choices from menus in $\mathcal{C}_1$. For any $A \in \mathcal{C}_1$, we define $\omega(A) = \succ_1$ and assign ownership to Agent~1. 
	\end{proof}
	
	\begin{reptheorem}{thm2}
		A choice function respects simple ownership if and only if it satisfies PC and AP.
	\end{reptheorem}
	
	\begin{proof}
		\emph{(Necessity)} Assume that the choice function $c$ respects simple ownership and, without loss of generality, let $\succ_1$ denote the order used over binary menus. PC holds because, if $c(\{x,y\})=x$ and $c(\{y,z\})=y$, then transitivity of $\succ_1$ implies $c(\{x,z\})=x$. Concerning AP, suppose by contradiction that $x \vartriangleright y$ and $y \vartriangleright x$. When prevalence holds, there exist sequences of alternatives and menus supporting it. In particular, because Condition~2 holds for $\succ_1$ over some menu $A_i$, it has to be $\omega(A_i)=\succ_2$. Hence, the chain supporting $x \vartriangleright y$ implies $x\succ_2 y$ and, similarly, the chain supporting $y \vartriangleright x$ implies $y\succ_2 x$, contradicting the asymmetry of $\succ_2$. 
		
		\medskip\noindent
		\emph{(Sufficiency)} Assume PC and AP. We construct $\succ_1,\succ_2$ and a simple ownership rule $\omega$ that are consistent with $c$. We break the argument into three steps: 1) define $\succ_1$ and $\succ_2$; 2) define $\omega$; 3) show that they rationalize $c$ via simple ownership.
		
		\smallskip\noindent
		\emph{Step 1}: Define a binary relation $\succ_1$ on $X$ by $x \succ_1 y$ if and only if $c(\{x,y\})=x$. Then $\succ_1$ is asymmetric and connected by construction. Its transitivity follows directly from PC: if $x\succ_1 y$ and $y\succ_1 z$, then $c(\{x,y\})=x$ and $c(\{y,z\})=y$; hence, PC implies $c(\{x,z\})=x$, i.e. $x\succ_1 z$. We conclude that $\succ_1$ is an order. 
		
		Define another binary relation $\succ_2^\ast$ by $x \succ_2^\ast y$ if and only if $x,y$ are in a menu $A$ with $c(A)=x$ and there exists $z\in A$ such that $c(\{x,z\})=z$. Intuitively, $x\succ_2^\ast y$ holds when $x$ is chosen from some menu $A$ containing $y$, but it loses a pairwise comparison within the same menu. We claim that $\succ_2^\ast$ is acyclic. By contradiction, suppose that there exists a directed cycle $y_1\succ_2^\ast y_2\succ_2^\ast\cdots\succ_2^\ast y_k\succ_2^\ast y_1$. By the definition of $\succ_2^\ast$, there exists a sequence of menus $A_1,\dots,A_k$ such that $c(A_i)=y_i$ and $y_{i+1} \in A_i$ for $i \in \left\{1,\dots,k\right\}$ (where $y_{k+1}=y_1$), with some $z_i\in A_i$ such that $c(\{y_i,z_i\})=z_i$. It follows that $y_1 \vartriangleright y_k$ and $y_k \vartriangleright y_1$, contradicting AP. Hence, $\succ_2^\ast$ contains no cycles. 
		
		The transitive closure $\succ_2^{\ast,\mathrm{tc}}$ of $\succ_2^\ast$ is acyclic and transitive. By Szpilrajn's Extension Theorem \citep{szpilrajn1930extension}, there exists an order $\succ_2$ that extends $\succ_2^{\ast,\mathrm{tc}}$.
		
		\smallskip\noindent
		\emph{Step 2.} Define the ownership rule by
		\[
		\omega(A)=
		\begin{cases}
			\succ_1, & \text{if } c(A)=\max(A, \succ_1),\\[4pt]
			\succ_2, & \text{otherwise.}
		\end{cases}
		\]
		Note that for any two-element menu $A=\{x,y\}$ we have $c(A)=\max(A, \succ_1)$ by the definition of $\succ_1$; thus, $\omega$ is indeed a simple ownership rule.
		
		\smallskip\noindent
		\emph{Step 3.} Finally, we show that for any menu $A$ we have $c(A)=\max(A, \omega(A))$. (If $A$ is a singleton, this is trivial; suppose $|A| \ge 2$.) If $c(A)=\max(A, \succ_1)$ then this holds by definition. Otherwise, let $c(A) = x_j \neq x_i=\max(A, \succ_1)$. We claim that $x_j=\max(A, \succ_2)$. Pick any $y\in A\setminus\{x_j\}$. Because $x_i=\max(A, \succ_1)$, we have $c(\{x_j,x_i\})=x_i$. By the definition of $\succ_2^\ast$ (replace $z$ with $x_i$), it follows from $c(A)=x_j$ and $x_i \in A$ that $x_j\succ_2^\ast y$. Because the transitive closure $\succ_2^{\ast,\mathrm{tc}}\subseteq\succ_2$, we have $x_j\succ_2 y$. Since this holds for any $y\in A\setminus\{x_j\}$, $x_j$ is the $\succ_2$-maximal element in $A$ and thus $c(A)=\max(A, \succ_2)$, establishing that $c(A)=\max(A, \omega(A))$. 
	\end{proof}
	
	\subsection{\sc \textbf{Proof for Section~\ref{sec:borda}}}\label{app:borda}
	
	\begin{reptheorem}{thm3}
		A choice correspondence respects Borda ownership if and only if it satisfies PC, ATT, EFF, EI, NBI, SI and EXP.
	\end{reptheorem}
	\begin{proof}
		
		\emph{(Necessity)} We show the necessity of each axiom, one by one. Under Borda ownership, $x = c(A)$ implies $x = \max(A, \succ_i)$ for some $i \in \{1,2\}$; without loss of generality, we assume $i=1$ and use $x = \max(A, \succ_1)$.
		
		\smallskip\noindent	
		\emph{(PC)} If $x \gg y$ and $y \gg z$, then $x \succ_i y \succ_i z$ for $i=1,2$. Because $\succ_i$ is an order for $i=1,2$, we have $x \succ_i z$ for $i=1,2$, and hence $x \gg z$.
		
		\smallskip\noindent	
		\emph{(ATT)} Let $x \in c(A)$ and $y \notin A$ with $x \gg y$. Without loss of generality, let $x = \max(A,\succ_1)$ and $t = \max(A,\succ_2)$, possibly with $x=t$. Because $x \gg y$ and $t \succ_2 y$, the tops in $A \cup \{y\}$ are still $x$ and $t$. Adding $y$ raises $b_2(x,\cdot)$ by one because $x \succ_2 y$, while it raises $b_1(t,\cdot)$ by one if $t \succ_1 y$ and leaves it unchanged otherwise. If $t \gg y$, both counts rise by one and $c(A \cup \{y\}) = c(A)$. If $t \not\gg y$, $b_1(t,\cdot)$ remains unchanged and $c(A \cup \{y\}) = \{x\} = c(A) \cap \{z \in A \mid z \gg y\}$.
		
		\smallskip\noindent		
		\emph{(EFF)} If $x \in c(A)$, then $x = \max(A, \succ_1)$. If there exists $y \in A$ with $y \gg x$, then in particular $y \succ_1 x$, violating the $\succ_1$- maximality of $x$. 
		
		\smallskip\noindent	
		\emph{(EI)} Suppose $c(\{a,b\}) = \{a,b\}$ for all $a,b \in \{x,y,z\}$. Assume without loss of generality $x \succ_1 y \succ_1 z$. Then $c(xy) = \{x, y \}$ implies $y \succ_2 x$; similarly, we obtain $z \succ_2 y \succ_2 x$. Because the chains under the two orders are opposite, exactly one of the three alternatives is an intermediate and $c(\{x,y,z\})$ is a doubleton.
		
		\smallskip\noindent	
		\emph{(NBI)} Suppose $x\langle y \rangle z$. Then $x \succ_{i} y \succ_{i} z$ and $z \succ_{3-i} y \succ_{3-i} x$ for some $i \in \{1,2\}$. If a menu $A$ contains $x,z$, then $y \neq \max(A, \succ_i)$ for $i=1,2$; that is, $y$ is never a top choice. Hence, $y \notin c(A)$.
		
		\smallskip\noindent	
		\emph{(SI)} Suppose $x\langle y \rangle z$. Then $x \succ_{i} y \succ_{i} z$ and $z \succ_{3-i} y \succ_{3-i} x$ for some $i \in \{1,2\}$. If $c(wy)=\{w,y\}$ for some $w \notin \{x,y,z\}$, then $y \succ_j w$ and $w \succ_{3-j} y$ for some $j \in \{1,2\}$. Suppose $j =i$. (We omit the similar argument for $j \neq i$.) By transitivity of $\succ_i$, we have $x \succ_{i} y \succ_{i} w$ and $w \succ_{3-i} y \succ_{3-i} x$; but this implies $c(xw)= \{x,w\}$, $c(wy)= \{w,y\}$ and $c(xwy)= \{x,w\}$ so $x\langle y \rangle w$ holds. 
		
		\smallskip\noindent		
		\emph{(EXP)} Without loss of generality, assume $c(A) = \{x,z\}$ with $x = \max(A, \succ_1)$ and $z = \max(A, \succ_2)$. The non-owner's counts for $x$ and $z$ in $A$ are equal.
		
		\smallskip	\noindent
		\emph{(i)} Add two alternatives $y,w \notin A$ with $x \gg y$, $z \gg w$, $x \| w$ and $y \| z$ to a menu $A$. This increases the non-owner's counts for $x$ and $z$ in $A \cup \{y,w\}$ exactly by one, so $c(A \cup \{y,w\}) = \{x,z\}$.
		
		\smallskip	\noindent
		\emph{(ii)} Let $y \notin A$ with $x \langle y \rangle z$. Because $x \succ_1 z \succ_1 y$ and $y \succ_2 z \succ_2 x$, adding $z$ leaves the counts for $x$ and $y$ in $A \cup \{z\}$ unchanged, so $c(A \cup \{z\}) = \{x,y\}$.
		
		\medskip\noindent
		\textit{(Sufficiency)} This is a long argument, and it is convenient to break it up into lemmas and propositions. For reference, these will be consecutively numbered \ref{thm3}.$x$ with $x=1,2, \ldots$. 
		
		The following terminology is useful. We say that a pair $\{x,y\}$ is a \emph{tie} if $c(\{x,y\})=\{x,y\}$, and a menu $A$ is \emph{all ties} if every pair in $A$ is a tie. By definition, $x \langle y \rangle z$ if $\{x,y,z\}$ is all ties and $c(\{x,y,z\})=\{x,z\}$. Moreover, an intermediate exists only if $\{x,y,z\}$ is all ties. 
		
		The first proposition states that the combination of four axioms (EI, NBI, SI, EXP) implies that non-resolute choices over pairs propagate to larger menus. Specifically, if every choice over pairs from a menu $A$ is a doubleton (i.e., it is non-resolute), then $c(A)$ must also be a doubleton. We call this property \emph{Non-resolute propagation} (NRP). 
		
		\medskip
		\noindent
		\emph{\textbf{(NRP)}:} If the menu $A$ is all ties, then $c(A)$ is a doubleton.
		\medskip
		
		We prove the following proposition.
		
		\begin{propnum}\label{prop31}
			EI, NBI, SI and EXP imply NRP.
		\end{propnum}
		\begin{proof}
			Consider a choice correspondence $c$ and an all ties menu $A$. Then EI implies that for any triplet $T \subseteq A$ we have $|c(T)|=2$; the unique element of $T\setminus c(T)$ is the intermediate for $T$, denoted $\iota (T)$. We say that an order $\gtrdot$ on a menu $A$ \emph{describes} the \emph{intermediates map} $\iota ( \cdot )$ if $\iota (T)$ is the $\gtrdot$-median of $T$ for any triplet $T \subseteq A$.
			
			Before completing the proof, we establish two lemmas. The first one claims that the intermediates in a four-element set can always be described by an order.
			
			\begin{lemmanum}\label{lem:quadruples}
				For every four-element menu $A$ that is all ties, there is an order $\gtrdot$ that describes the intermediates map on the triplets of $A$.
			\end{lemmanum}		
			\begin{proof}
				First, we claim that any intermediate of a triplet in $A$ is also an intermediate for at least one other triplet. Suppose $A = \{ t, u, v, w \}$ with $t \langle u \rangle v$. Because $A$ is all ties, we have $c(\{u,w\})=\{u,w\}$; then SI yields either $t \langle u \rangle w$ or $v \langle u \rangle w$ and the claim holds.
				
				Second, a four-element menu $A$ has four distinct triplets and, as argued above, each triplet $T$ has an intermediate $\iota (T)$. The claim implies that $A$ has exactly two elements (denote them $b,c$) that are intermediates of two triplets each, and two other elements (denote them $a,d$) that are never intermediates. 
				
				Each triplet in $A$ omits exactly one element. The triplet omitting $b$ contains $c$, so $a \langle c \rangle d$; symmetrically, $a \langle b \rangle d$. The other two triplets $\{a,b,c\}$ and $\{b,c,d\}$ have distinct intermediates in $\{b,c\}$. If $a \langle b \rangle c$ and $b \langle c \rangle d$, then $\iota (\cdot)$ is described by $a \gtrdot b \gtrdot c \gtrdot d$; and if $a \langle c \rangle b$ and $c \langle b \rangle d$, then $\iota (\cdot)$ is described by $a \gtrdot c \gtrdot b \gtrdot d$.\end{proof}
			
			The second lemma proves that the orders that describe locally the intermediates map for the triplets of an all ties menu $A$ can be glued into a single order for $A$ itself.
			
			\begin{lemmanum}[Gluing]\label{lem:gluing}
				Let $\iota$ be the intermediate map for the triplets of an all ties menu $A$. Suppose that for every four-element subset $B$ of $A$ the restriction of $\iota$ to the triplets in $B$ is described by an order on $B$. Then $\iota$ is described by an order on $A$.
			\end{lemmanum}
			\begin{proof}
				The proof proceeds by induction on the size of $A$. The base case $|A|=4$ is proved in Lemma~\ref{lem:quadruples} above. Let $|A|=n \geq 5$ and suppose that the claim is true for all menus of size at most $n-1$. Pick an arbitrary $a \in A$, and let $x_1 \gtrdot \cdots \gtrdot x_{n-1}$ describe $\iota$ on $A\setminus\{a\}$ by our induction hypothesis. For each $i \in \{2,\dots,n-1\}$, exactly one of the following holds: (I)~$\iota (\{x_1,a,x_i\})=x_1$; or (J)~$\iota (\{x_1,a,x_i\})=a$; or (K)~$\iota (\{x_1,a,x_i\})=x_i$.
				
				Suppose (I) holds for some $i$, and let $j\in\{2,\dots,n-1\}\setminus\{i\}$. For $B=\{x_1,x_i,x_j,a\}$, the order $\gtrdot$ has $x_1 \gtrdot x_i$ (and $x_1 \gtrdot x_j$). Then $\iota (\{x_1,a,x_i\})=x_1$ forces $a \gtrdot x_1$. Moreover, because $\iota (\{x_1,x_i,x_j\}) \neq x_1$, we have $\iota (\{x_1,a,x_j\})=x_1$ because by Lemma~\ref{lem:quadruples} each alternative must be an intermediate for either zero or two triplets. Hence, (I) holds for all indices and inserting $a \gtrdot x_1$ describes $\iota$ on $A$.
				
				Suppose now that (I) holds for no index. Let $J$ and $K$ respectively denote the sets of indices for which (J) and (K) occur, so that $J\cup K=\{2,\dots,n-1\}$. If $i \in J$, then $x_1\gtrdot a \gtrdot x_i$; when $x_i \gtrdot x_j$, we have $x_1\gtrdot a \gtrdot x_i \gtrdot x_j$, so $\iota (\{x_1,a,x_j\})=a$ and $j\in J$. Hence, $J$ is an upper interval of indices and then $K = \{2,\dots,p\}$ is the complementary initial segment. ($K$ can be empty, for $p=1$; or $K$ can include all indices, for $p=n-1$.) Insert $a$ between $x_p$ and $x_{p+1}$ (or after $x_{n-1}$ if $p=n-1$). 
				
				It remains to check that the order $x_1 \gtrdot \cdots \gtrdot x_p \gtrdot a \gtrdot x_{p+1} \gtrdot \cdots \gtrdot x_{n-1}$ describes $\iota$ on the triplets $\{a,x_i,x_j\}$ with $1<i<j$. Given a 4-element menu $\{a, x_1, x_i, x_j\}$, there are three possible cases: (i) if $i,j \in J$, then $\iota (\{x_1,a,x_i\})=a$. Since $\iota (\{x_1,x_i,x_j\})= x_i$, then $\iota (\{a,x_i,x_j\})= x_i$ because each alternative is an intermediate of either zero or two triplets. Similarly, (ii) if $i,j \in K$, then by a symmetric argument $\iota(\{a,x_i,x_j\})= x_j$. Finally, (iii) if $i \in K$ and $j \in J$, then $\iota (\{x_1, x_i, x_j \}) = \iota (\{x_1,x_i, a\})= x_i$, and $ \iota(\{x_1,a,x_j\})=a$. Therefore, $\iota (\{a,x_i,x_j\})=a$ by the usual multiplicity property of intermediates. This completes the proof of Lemma~\ref{lem:gluing}.
			\end{proof}
			
			We return to the proof of Proposition~\ref{prop31}. If $|A|=2$, then NRP holds trivially; if $|A|=3$, EI implies $|c(A)|=2$. We prove by induction on $n=|A|\geq 3$ the following claim: given the order $x_1\gtrdot \cdots \gtrdot x_n$ that describes $\iota (\cdot)$ on $A$, $c(A)=\{x_1,x_n\}$ is a doubleton. (Lemma~\ref{lem:quadruples} and~\ref{lem:gluing} prove that such an order exists.)
			
			For $n=3$, this follows from the definition of an intermediate. Assume $n \geq 4$. For any $1<i<n$, the order implies $x_1 \langle x_i \rangle x_n$; because $x_1,x_n \in A$, NBI implies $x_i \notin c(A)$ and thus $c(A) \subseteq
			\{x_1, x_n\}$. Next, let $B=A\setminus\{x_2\}$. Because $B$ is all ties and $\gtrdot$ describes $\iota (\cdot)$ on $B$, the induction hypothesis gives $c(B)=\{x_1,x_n\}$. Then $x_1\langle x_2 \rangle x_n$ and $x_2\notin B$, so EXP.ii implies $x_1,x_n\in c(B\cup\{x_2\})=c(A)$. This yields $c(A)=\{x_1, x_n\}$ and completes the proof of Proposition~\ref{prop31}.
		\end{proof}
		
		The proof of sufficiency is by induction on the size of choice problems. It consists of three broad steps. First, we inductively construct a pair of connected preferences using the technique in~\citet{de2012reason}. Second, we prove that these two preferences are transitive, and therefore they are orders. Third, we show that these orders rationalize the choice correspondence under Borda ownership. 
		
		For greater readability, we use the notation $x \gg y$ when $c(xy)=\{x\}$, $x \| y$ when they are incomparable, and $x \langle y \rangle z$ when $\{x,y,z\}$ is all ties and $c(xyz)=\{x,z\}$. Intuitively, we prove that $x \gg y$ implies both $x \succ_1 y$ and $x \succ_2 y$; whereas $x \langle y \rangle z$ implies $x \succ_i y \succ_i z$ and $z \succ_{3-i} y \succ_{3-i} x$ for either $i=1$ or $i=2$, but not both.
		
		\smallskip\noindent	
		\textit{Step 1.} To set the background, define for each integer $k$ the subset of $X$
		\[
		U^k = \left\{ x \in X \setminus \left( \bigcup_{j=0}^{k-1} U^j \right) \middle|\ \nexists \ y \in X \setminus \left( \bigcup_{j=0}^{k-1} U^j \right) \text{ such that } y \gg x \right\},
		\]
		with $U^0 = \emptyset$. The subset $U^1$ consists of the undominated alternatives in $X$. For $k>1$, the subset $U^k$ consists of the alternatives in $X$ that are undominated after removing $\left( \bigcup_{i=0}^{k-1} U^i \right)$. The collection $\{U^k\}$ is a finite partition of $X$ into layers of alternatives that are undominated at different levels $k$. 
		
		Next, we partition each undominated layer $U^k$ into sublayers. For each integer $i$ define the subset of $U^k$: 
		\[
		E^{k,i} = \left\{ y \in U^k \setminus \left( \bigcup_{j=0}^{i-1} E^{k,j} \right) \middle|\ \nexists \ x,z \in U^k \setminus \left( \bigcup_{j=0}^{i-1} E^{k,j} \right)\text{ such that } x \langle y \rangle z \right\},
		\]
		with $E^{k,0} = \emptyset$. The sublayer $E^{k,1}$ contains the alternatives that are never intermediates between elements of $U^k$: we interpret them as the \emph{extreme} elements for $U^k$. Similarly, $E^{k,i}$ consists of the extreme elements in the sublayer $U^k \setminus \left( \bigcup_{j=0}^{i-1} E^{k,j} \right)$. Because $U^k$ is all ties, EI implies that, if $E^{k,i}$ contains two extremes $x,z$, any third element $y$ must be an intermediate. Consequently, $|E^{k,i}|\le 2$ with equality holding when $U^k \setminus \left( \bigcup_{j=0}^{i-1} E^{k,j} \right)$ is not empty and has at least two elements. In sum, each layer $U^k$ is partitioned into a collection of sublayers $\{ E^{k,i} \}$ each of which consists of one or two elements. 
		
		We use the collections $\{U^k\}$ and $\{E^{k,i}\}$ to assign relative ranks to alternatives in the inductive construction of a pair $\succ = (\succ_1,\succ_2)$ of binary relations. Afterwards, we will prove that the choice correspondence $c_\succ ( \cdot )$ generated by the pair $\succ$ under Borda ownership satisfies $c(A) = c_{\succ}(A)$ for any menu $A$.
		
		The construction of $\succ$ starts with a pair of alternatives in $X$ and incrementally adds one element at a time until $X$ is exhausted. By definition each alternative $x$ belongs to a unique sublayer $E^{k,i}$. Let $k(x), i(x)$ respectively denote the two integers associated with $x$. We add $x$ before $y$ if the ordered pair $\left( k(x), i(x) \right)$ is lexicographically inferior to $\left( k(y), i(y) \right)$. This convention does not uniquely enumerate all the alternatives, because $E^{k,i}$ usually contains two elements. However, the exact order in which two elements of $E^{k,i}$ are included is inconsequential.
		
		The construction proceeds by induction. For the base case, let $x$ and $y$ be the first two alternatives (in the lexicographic order defined above). If $x \gg y$, set $x \succ_1 y$ and $x \succ_2 y$; if $y \gg x$, set $y \succ_1 x, y \succ_2 x$; and if $ x \| y$, set $x \succ_1 y$ and $y \succ_2 x$. By construction, $c (\{x,y\}) = c_\succ (\{x,y\})$.
		
		For the inductive step, suppose that two orders $\succ_1,\succ_2$ have been defined on a subset $A \subset X$ and that $c_{\succ}(T) = c(T)$ for all $T \subseteq A$. Let $w \notin A$ be the next element to be added. We extend the pair $\succ = (\succ_1,\succ_2)$ on $A$ to the pair $\succ^\prime = (\succ_1^\prime, \succ_2^\prime)$ over $A \cup \{w\}$.
		
		We partition $A$ into the two disjoint subsets
		\begin{equation}\label{eq:BandC}
			B = \left\{x \in A \mid x \gg w \right\} \quad\mbox{ and } \quad
			C = \left\{x \in A \mid x \| w \right\}
		\end{equation}
		Note $B \cup C = A$, because under the lexicographic order there cannot exist $x \in A$ such that $w \gg x$.
		
		For each $x \in B$, define $x \succ_1^\prime w$ and $x \succ_2^\prime w$. If $x \in C$, we need to distinguish two cases. In the first case, $w$ is an intermediate in $C$: that is, there exist $z,z^\prime \in C$ such that $z \langle w \rangle z^\prime$. If there exists $y \in C$ with $x \succ_1 y$ and $x \langle w \rangle y$, we set $x \succ_1^\prime w$ and $w \succ_2^\prime x$; otherwise, there exists $y \in C$ with $x \succ_2 y$ and $x \langle w \rangle y$ and we set $x \succ_2^\prime w$ and $w \succ_1^\prime x$ . The following lemma confirms that this definition is well posed.
		
		\begin{lemmanum}\label{lemma3-4}
			Let $x \in C$. (i) If $z \langle w \rangle z^\prime$ for some $z,z^\prime \in C$, then there exists $y \in C$ with $x \langle w \rangle y$. (ii) Moreover, if there are $y,y^\prime \in C$ such that $x \langle w \rangle y$ and $x \langle w \rangle y^\prime$, then $x \succ_i y$ if and only if $x \succ_i y^\prime$ for both $i=1,2$.
		\end{lemmanum}
		\begin{proof}
			(i)If $x \in \{z,z^\prime\}$, let $y$ be the other element. Otherwise, $x \notin \{z,w,z^\prime\}$ and SI implies either $x \langle w \rangle z$ or $x \langle w \rangle z^\prime$. Choosing $y$ equal to either $z$ or $z^\prime$ establishes the claim.
			
			\noindent
			(ii) Assume $y,y^\prime \in C$ such that $x \langle w \rangle y$ and $x \langle w \rangle y^\prime$. Because $\{x,w,y\}$ and $\{x,w,y^\prime\}$ are all ties, $x \| y$ and $x \| y^\prime$. Suppose $x \succ_1 y$ and by contradiction $y^\prime \succ_1 x$. By the induction hypothesis, $y \succ_2 x$ and $x \succ_2 y^\prime$. Chaining the binary relations, we have $y^\prime \succ_1 x \succ_1 y$ and $y \succ_2 x \succ_2 y^\prime$ and thus $y \langle x \rangle y^\prime$. Because $x \| w$, SI implies either $y \langle x \rangle w$ or $w \langle x \rangle y^\prime$. But the former contradicts $x \langle w \rangle y$ and the latter contradicts $x \langle w \rangle y^\prime$.
		\end{proof}
		
		In the second case, $w$ is not an intermediate: there do not exist $z,z^\prime \in C$ with $z\langle w \rangle z^\prime$. If there exist $\xi \in B$ and $y \in C$ such that $y \succ_1 \xi$ (here, $y$ may be $x$ itself), we set $x \succ_1^\prime w$ and $w \succ_2^\prime x$; if instead there exist $\xi \in B$ and $y \in C$ with $y \succ_2 \xi$, we set $x \succ_2^\prime w$ and $w \succ_1^\prime x$; finally, if neither holds (i.e., there is no $\xi \in B$ and no $y \in C$ with either $y \succ_1 \xi$ or $y \succ_2 \xi$), then one is free to choose either definition (i.e., $x \succ_1^\prime w$ and $w \succ_2^\prime x$ for all $x \in C$, or $w \succ_1^\prime x$ and $x \succ_2^\prime w$ for all $x \in C$). The next lemma confirms that this definition is well posed.
		
		\begin{lemmanum}\label{lemma3-5}
			If there exist no $z,z^\prime \in C$ with $z\langle w \rangle z^\prime$, then there exist no $\xi,\xi^\prime \in B$ and $y,y^\prime \in C$ (possibly $y=y^\prime$) with $y \succ_2 \xi$ and $y^\prime \succ_1 \xi^\prime$.
		\end{lemmanum}
		\begin{proof}
			Assume by contradiction that there exist $\xi,\xi^\prime \in B$ and $y,y^\prime \in C$ with $y \succ_2 \xi$ and $y^\prime \succ_1 \xi^\prime$. Then $c(\{\xi,y\}) \neq \{\xi\}$ because $\xi \not\gg y$. Also, $c(\{\xi,y\}) \neq \{y\}$, or else $y \gg \xi$ which (by $\xi \in B$ and PC) would entail $y \gg w$, contradicting $y \in C$. Hence $c(\{\xi,y\}) = \{\xi, y\}$ and, by the same argument, $c(\{\xi^\prime,y^\prime \}) = \{\xi^\prime,y^\prime \}$. Therefore, under our induction hypothesis, we have
			\begin{equation}\label{eq:1}
				\left( \xi \succ_1 y \mbox{ and } y^\prime \succ_1 \xi^\prime \right) \quad\mbox{ and }\quad \left( y \succ_2 \xi \mbox{ and } \xi^\prime \succ_2 y^\prime \right)
			\end{equation}
			If $y=y^\prime$, then~(\ref{eq:1}) implies $\xi \langle y \rangle \xi^\prime$, with $\xi \neq \xi^\prime$. Because $y \in C$, $c(\{y,w\})=\{y,w\}$ and then by SI either $\xi \langle y \rangle w$ or $\xi^\prime \langle y \rangle w$, which is impossible because $\xi \gg w$ and $\xi^\prime \gg w$. Hence, $y \neq y^\prime$ from now on.
			
			We can distinguish two cases.
			
			\smallskip\noindent	
			\textit{Case 1:} $\xi \gg y^\prime$ and $\xi^\prime \gg y$. 
			By dominance, $\xi \succ_1 y^\prime$ and $\xi^\prime \succ_1 y$. Combined with~(\ref{eq:1}), we have $\xi \succ_1 y^\prime \succ_1 \xi^\prime \succ_1 y$. Similarly, we get $\xi^\prime \succ_2 y \succ_2 \xi \succ_2 y^\prime$. By the induction hypothesis, $\succ_1$ and $\succ_2$ are transitive on $A$ and $c = c_\succ$ for all menus in $A$; therefore, we can check that $c(\{y,y^\prime\}) = \{y,y^\prime\}$ and, similarly, that the three menus $\{\xi,y\}$, $\{\xi^\prime,y^\prime\}$, $\{\xi, \xi^\prime\}$ are ties.
			
			Because the menu $\{y,y^\prime,w\}$ is all ties, the choice $c(\{y, y^\prime, w \})$ must be a doubleton. Suppose $c(\{y,y^\prime,w\}) = \{y,w\}$. Then $y\langle y^\prime \rangle w$ and, because $\{\xi^\prime,y^\prime\}$ is a tie, SI implies either $y \langle y^\prime \rangle \xi^\prime$ or $\xi^\prime \langle y^\prime \rangle w$. But this is impossible because $\xi^\prime \gg y$ and $\xi^\prime \gg w$. A similar argument rules out $c(\{y,y^\prime,w\}) = \{y^\prime,w\}$. The third and last possibility is $c(\{y,y^\prime,w\}) = \{y,y^\prime\}$, which contradicts the assumption of the lemma that there are no $z,z^\prime \in C$ with $z\langle w \rangle z^\prime$. We conclude that the first case cannot occur.
			
			\smallskip\noindent
			\textit{Case 2:} either $\xi \not\gg y^\prime$ or $\xi^\prime \not\gg y$ (or both).
			Without loss of generality, suppose by contradiction that $c(\{\xi,y^\prime\}) \neq \{\xi\}$. (A similar argument holds if $c(\{\xi^\prime,y\}) \neq \{\xi^\prime\}$.) Then either $c(\{\xi,y^\prime\}) = \{y^\prime\}$ or $c(\{\xi,y^\prime\}) = \{\xi,y^\prime\}$. We show that neither is possible, reaching a contradiction.
			
			The first case $c(\{\xi,y^\prime\}) = \{y^\prime\}$ is not possible, because PC would imply $ y^\prime \gg w$, contradicting the assumption $y^\prime \in C$. Consider the other case $c(\{\xi,y^\prime\}) = \{\xi,y^\prime\}$. If $\xi \succ_1 y^\prime$, then $c(\{\xi,y^\prime\}) = \{\xi,y^\prime\}$ implies $y^\prime \succ_2 \xi$. Using (\ref{eq:1}) and transitivity, we have the two chains $\xi \succ_1 y^\prime \succ_1 \xi^\prime$ and $\xi^\prime \succ_2 y^\prime \succ_2 \xi$, which give $\xi \langle y^\prime \rangle \xi^\prime$. Because $y^\prime \in C$, we have that $c(\{y^\prime,w\}) = \{y^\prime,w\}$. Then SI implies either $\xi \langle y^\prime \rangle w$ or
			$\xi^\prime \langle y^\prime \rangle w$. But neither is possible because $c(\{\xi,w\}) = \{\xi\}$ and $c(\{\xi^\prime,w\}) = \{\xi^\prime\}$. This rules out $\xi \succ_1 y^\prime$.
			
			Under the current assumption $c(\{\xi,y^\prime\}) = \{\xi,y^\prime\}$, therefore, it must be $y^\prime \succ_1 \xi$ and $\xi \succ_2 y^\prime$. Combining with (\ref{eq:1}), we get $y^\prime \succ_1 \xi \succ_1 y$ and $y \succ_2 \xi \succ_2 y^\prime$. This implies that $y \| y^\prime$; then, because $y,y^\prime \in C$, EI implies that $c(\{y,y^\prime,w\})$ must be a doubleton. This may happen in three different ways. We show that they are all impossible, reaching the final contradiction.
			
			First, suppose $c(\{y,y^\prime,w\}) = \{y,w\}$. Then $y \langle y^\prime \rangle w$. Because $c(\{\xi,y^\prime\}) = \{\xi, y^\prime\}$, SI implies $y \langle y^\prime \rangle \xi$ or $\xi \langle y^\prime \rangle w$. But the first case is not possible because the two chains imply $c(\{y^\prime,\xi,y\}) = \{y,y^\prime\}$, and the second case is not possible because $\xi \in B$. A similar argument rules out $c(\{y,y^\prime,w\}) = \{y^\prime,w\}$. The last possibility
			is $c(\{y,y^\prime,w\}) = \{y,y^\prime\}$, which contradicts the assumption of the lemma that there are no $z,z^\prime \in C$ with $z\langle w \rangle z^\prime$.
		\end{proof}
		
		\smallskip\noindent	
		\textit{Step 2.} 	
		We have defined a pair of connected preferences $\succ_1^\prime, \succ_2^\prime$ over $A \cup \{w\}$ that extend the orders $\succ_1, \succ_2$ over $A$. The second step proves that at each step of our inductive construction the extensions $\succ_1^\prime$ and $\succ_2^\prime$ are transitive on $A \cup \{w\}$. 
		
		Given that the statement of transitivity has two premises, we arrange the proof in two lemmas. The first lemma covers the two cases where one of the premises uses $\succ$, and hence draws directly on the inductive hypothesis. The second lemma deals with the case where neither premise uses $\succ$.
		
		\begin{lemmanum}\label{lemma3-6}
			Let $(\succ_1, \succ_2)$ be two orders defined over a subset $A$ of $X$ such that $c = c_\succ$ for all $T \subseteq A$, under Borda ownership. For the next element $w \notin A$ to be added, let $(\succ_1^\prime, \succ_2^\prime)$ be the extensions of $(\succ_1, \succ_2)$ as defined above. Let $x,y$ be two elements of $A$. For $i \in \{1,2\}$, the following two statements hold:\\
			\hspace*{2cm} 1) If $x \succ_i y$ and $y \succ_i^\prime w$, then $x \succ_i^\prime w$.\\ 
			\hspace*{2cm} 2) If $w \succ_i^\prime y$ and $y \succ_i x$, then $w \succ_i^\prime x$.
		\end{lemmanum}
		\begin{proof}
			The proofs of the two statements are similar. We prove only the first.
			
			If $x \in B$, then $x \gg w$ so by construction $x \succ_i^\prime w$ for $i \in \{1,2\}$. This establishes the claim. From now on, assume $x \in C$. 
			
			We distinguish two cases, depending on whether $w$ is an intermediate in $C$. For the first case, suppose there are no $z,z^\prime \in C$ with $z \langle w \rangle z^\prime$. If $y \in B$ (and by assumption $x \in C$), the construction in the paragraph preceding the statement of Lemma~\ref{lemma3-5} sets $x \succ^\prime_i w$ if and only if $x \succ_i y$, as desired. If instead $y \in C$, then the same construction and Lemma~\ref{lemma3-5} show that $y \succ_i^\prime w$ if and only if $x \succ_i^\prime w$, as desired.
			
			For the second case, suppose there exist $z,z^\prime \in C$ with $z \langle w \rangle z^\prime$. By Lemma~\ref{lemma3-4}.(i), there exists $x^\prime \in C$ with $x \langle w \rangle x^\prime$. The construction preceding Lemma~\ref{lemma3-4} sets $x \succ_i^\prime w$ if and only if $x \succ_i x^\prime$. So, it suffices to prove that, if $x \succ_i y$ and $y \succ^\prime_i w$, then $x^\prime \succ_i x$ is impossible. By contradiction, assume $x^\prime \succ_i x$ (and $x \succ_{3-i} x^\prime$).
			
			If $y \in B$ (i.e., $y \gg w$) then $c(\{x,y\}) \neq \{x\}$, or otherwise $x \gg y$ and PC would imply $x \gg w$ (i.e., $x \in B$). Moreover, because $x \succ_i y$, $c(\{x,y\}) \neq \{y\}$. Therefore, we must have $c(\{x,y\}) = \{x,y\}$ and $y \succ_{3-i} x$. Combined with the premise $x^\prime \succ_i x$ and $x \succ_{3-i} x^\prime$, the transitivity of $\succ_i, \succ_{3-i}$ on $A$ implies $x^\prime \succ_i x \succ_i y$ and $y \succ_{3-i} x \succ_{3-i} x^\prime$, so that $y \langle x \rangle x^\prime$ holds. Because $c(\{x,w\}) = \{x,w\}$, SI implies either $y \langle x \rangle w$ or $w \langle x \rangle x^\prime$. The former is not possible because $c(\{y,w\}) = \{y\}$, and the latter contradicts $x \langle w \rangle x^\prime$. For $y \in B$, a contradiction is reached.
			
			Now, suppose $y \in C$; then $y \succ_i^\prime w$ if and only if $w \succ_{3-i}^\prime y$ by construction. Consider $c(\{ x^\prime,y\})$. We must have $c(\{ x^\prime,y \}) \ne \{y\}$, or otherwise $y \succ_i x^\prime$ which, together with $x^\prime \succ_i x$ and $x \succ_i y$, would violate transitivity on $A$. 
			
			We claim that $c(\{x^\prime,y\}) \neq \{x^\prime\}$. Assume by contradiction $x^\prime \gg y$. Then $x \succ_{3-i} x^\prime \succ_{3-i} y$ and $x \succ_i y$ give $x \gg y$. Because $y \in C$, $c(\{w,y\}) = \{w,y\}$. Because $x \langle w \rangle x^\prime$, we have $c(\{x,w\}) = \{x,w\}$ and $c(\{x^\prime,w\}) = \{x^\prime,w\}$. Then ATT would imply $c(\{x,y,w\}) = \{x\}$ and $c(\{x^\prime,y,w\}) = \{x^\prime\}$, while SI would imply that either $x \langle w \rangle y$ or $y \langle w \rangle x^\prime$. The two implications are irreconcilable, so the claim holds.
			
			Consider $c(\{x^\prime,y\}) = \{x^\prime,y\}$. Combining $x \succ_i y$ and $x^\prime \succ_i x$, transitivity yields $x^\prime \succ_i y$. We claim that $c(\{x,y\}) \ne \{x\}$. As before, $x \langle w \rangle x^\prime$ implies either $x \langle w \rangle y$ or $y \langle w \rangle x^\prime$. The former is not possible because $c(\{x,y\}) = \{x\}$; and the latter implies $c(\{y, w, x^\prime\}) = \{x^\prime,y\}$, which rules out $y \succ_i^\prime w$ because $x^\prime \succ_i y$.
			
			Because $x \succ_i y$, $c(\{x,y\}) \ne \{y\}$. We are left with $c(\{x,y\}) = \{x,y\}$, that yields
			$x^\prime \succ_i x \succ_i y$ and $y \succ_{3-i} x \succ_{3-i} x^\prime$, and thus $y \langle x \rangle x^\prime$. Because $c(\{x,w\}) = \{x,w\}$, SI implies either $y \langle x \rangle w$ or $w \langle x \rangle x^\prime$. The latter contradicts $x \langle w \rangle x^\prime$, so we assume $y \langle x \rangle w$. Meanwhile, by SI, $x \langle w \rangle x^\prime$ and $c(\{y,w\}) = \{y,w\}$ imply either $x \langle w \rangle y$ or $y \langle w \rangle x^\prime$. The former contradicts $y \langle x \rangle w$, so we maintain $y \langle w \rangle x^\prime$. Because $x^\prime \succ_i y$, our construction sets $x^\prime \succ_i^\prime w \succ_i^\prime y$ ruling out $y \succ_i^\prime w$, which is a contradiction. This exhausts all possibilities, proving the claim by contradiction.
		\end{proof}
		
		We proceed to prove the second lemma, covering the third and last case left out by Lemma~\ref{lemma3-6}.
		
		\begin{lemmanum}\label{lemma3-7}
			Under the same assumptions of Lemma~\ref{lemma3-6}, the following third statement holds:\\
			\hspace*{2cm} 3) If $x \succ_i^\prime w$ and $w \succ_i^\prime y$, then $x \succ_i y$.
		\end{lemmanum}
		\begin{proof}
			If $c(\{x,y\}) = \{x\}$, then we are done by definition of $\succ_i$. So, assume $c(\{x,y\}) \neq \{x\}$. Note that $w \succ_i^\prime y$ rules out $y \in B$: we must have $y \in C$. We distinguish two cases.
			
			The first case has $x \in B$ and $y \in C$. We can rule out $c(\{x,y\}) = \{y\}$ because otherwise PC would imply $y \succ_i^\prime w$. Hence, we have $c(\{x,y\}) = \{x,y\}$. Assume by contradiction $y \succ_i x$ and thus, by the induction hypothesis, $x \succ_{3-i} y$.
			
			Assume that there exists no $z \in C$ with $y \langle w \rangle z$. The construction in the paragraph preceding the statement of Lemma~\ref{lemma3-5} sets $y \succ^\prime_i w$, contradicting the initial assumption $w \succ_i^\prime y$.
			
			Assume now that there exists $z \in C$ with $y \langle w \rangle z$. Because $w \succ_i^\prime y$, the construction preceding Lemma~\ref{lemma3-4} sets $z \succ_i y$ and $y \succ_{3-i} z$. Combined with $y \succ_i x$ and $x \succ_{3-i} y$, we have $z \succ_i y \succ_i x$ and $x \succ_{3-i} y \succ_{3-i} z$ by the transitivity of $\succ_i$ and $\succ_{3-i}$ entailed from our induction hypothesis; hence, $x \langle y \rangle z$. Because $c(\{y,w\}) = \{y,w\}$, SI implies either $x \langle y \rangle w$ or $w \langle y \rangle z$. The former is not possible because $x \in B$ and the latter is ruled out by the premise $y \langle w \rangle z$. Therefore, $y \succ_i x$ cannot hold and we get $x \succ_i y$.
			
			The second case has $x \in C$ and, as noted earlier, $y \in C$. Suppose that $w$ is not revealed to be intermediate in $C$; that is, that there exist no $x^\prime,y^\prime \in C$ such that $x \langle w \rangle x^\prime$ and $y \langle w \rangle y^\prime$. (Here, $x^\prime$ could be $y$ and $y^\prime$ could be $x$.) Then the construction preceding Lemma~\ref{lemma3-5} implies that only one of $x \succ_i^\prime w$ and $w \succ_i^\prime y$ holds. So transitivity holds vacuously. 
			
			Suppose now that $w$ is revealed to be intermediate in $C$, and there exist $x^\prime,y^\prime \in C$ such that $x \langle w \rangle x^\prime$ and $y \langle w \rangle y^\prime$. Because the lemma assumes $x \succ^\prime_{i} w$ and $w \succ^\prime_{i} y$, by construction we have $x \succ_i x^\prime$ and $y \succ_{3-i} y^\prime$. Note that $x^\prime = y$ and $y^\prime = x$ is possible. If $x^\prime =y$, we are done. If not, SI implies either $x \langle w \rangle y$ or $y \langle w \rangle x^\prime$. If $x \langle w \rangle y$, then $x \succ_{i} y$ by Lemma \ref{lemma3-4}.ii. If $y \langle w \rangle x^\prime$, then $w \succ_{i} y$ implies $x^\prime \succ_{i} y$ by construction. Because $x \succ_{i} x^\prime$, transitivity of $\succ_{i}$ yields $x \succ_{i} y$.
		\end{proof}
		
		We have inductively defined a pair $\succ^\prime = \left(\succ_1^\prime, \succ_2^\prime \right)$ of orders. It remains to show that this pair $\succ^\prime$ rationalizes the choice correspondence $c$ under Borda ownership; that is, $c_{\succ^\prime}(A) = c(A)$ for any menu $A$. For accuracy, we use the notation $b_i (x, A; \succ)$ to pinpoint that Borda counts are taken with respect to an order $\succ$.
		
		\smallskip\noindent
		\textit{Step 3.} We proceed by induction using the lexicographic order for adding alternatives defined above. The base case ($|A|=2$) is trivial. The inductive hypothesis is $c_{\succ^\prime}(T) = c(T)$ for all $T \subseteq A$. We choose $w \notin A$ following the lexicographic order and show that $c_{\succ^\prime}(T \cup \{w\}) = c(T \cup \{w\})$ for any $T$. Because $|c(T)| \leq 2$, we consider two cases: (i)~$|c(T)| = 1$, and (ii)~$|c(T)|=2$. Let $B$ and $C$ be defined as in Equation~(\ref{eq:BandC}).
		
		\medskip\noindent
		\textbf{Case (i): $|c(T)| = 1$.} We consider two subcases, depending on whether $\max(T, \succ_1) = \max(T,\succ_2)$, or not.
		
		\smallskip\noindent
		\textbf{Subcase (i.1) $x = \max(T, \succ_1) = \max(T, \succ_2)$}. If \emph{$x \in B$}, EFF implies $c(T \cup \{w\}) = \{x\}$ because $x \gg y$ for all $y \in T \cup \{w\}$; hence, by construction, $\max(T\cup \{w\}, \succ_1^\prime) = \max(T\cup \{w\},\succ_2^\prime) = x$ and we get $\{ x \} = c_{\succ^\prime}(T \cup \{w\})$. 
		
		If \emph{$x \in C$}, then $c(\{x,w\}) = \{x,w\}$. Because of the lexicographic order, $c(\{w,y\}) \neq \{w\}$ for any $y \in T$ and thus $w$ is an extreme. Also, by assumption, $x = \max(T, \succ_i)$ for $i=1,2$: so, $x \gg y$ for any $y \ne x$ and $y \in T$. Starting with $\{ x,w \}$ and $c(\{ x,w \})= \{ x,w \}$, the repeated application of ATT yields $c(T \cup \{w\}) = \{x\}$. Returning to $c(\{x,w\}) = \{x,w\}$, this implies $x = \max(T\cup \{w\}, \succ_i^\prime)$ and $w = \max(T \cup \{w\}, \succ_{3-i}^\prime)$ for some $i \in \{1,2\}$. Then $y \succ_{i} w$ for all $y \in T$, or else we would have $w \gg y$ for some $y \in T$, violating the lexicographic order for adding alternatives. Recall $x = \max(T, \succ_{3-i})$: this yields $x \succ_{3-i} y$ for all $y \in T$, from which we get $b_{3-i}(x, T \cup \{w\}; \succ^\prime ) > b_{i }(w, T \cup \{w\}; \succ^\prime)$ and then $c_{\succ^\prime}(T \cup \{w\}) = \{ x \}$.
		
		\smallskip\noindent
		\textbf{Subcase (i.2) $x = \max(T, \succ_i) \ne \max(T, \succ_{3-i}) = y$.} Clearly, $x \succ_i y$ and $y \succ_{3-i} x$. Suppose $c_{\succ}(T) = \{x\}$. Then $\{x\} = c(T)$ by the induction hypothesis, with $b_{3-i}(x,T; \succ) > b_{i}(y, T; \succ)$. If \emph{$x \in B$}, then ATT (applied on $c(T)$ to $c(T \cup \{w\})$) implies $c(T \cup \{w\}) = \{x\}$; and, by construction, $x \succ_i^\prime w$ and $y \succ_{3-i}^\prime x \succ_{3-i}^\prime w$. Therefore, $\max(T \cup \{w\}, \succ_i^\prime) = x$ and $\max(T \cup \{w\}, \succ_{3-i}^\prime) = y$. Moreover, $b_{3-i}(x,T \cup \{w\}; \succ^\prime) > b_{i}(y, T \cup \{w\}; \succ^\prime)$. This implies $\{x\} = c_{\succ^\prime}(T \cup \{w\})$. If \emph{$x \in C$}, then by construction, $x \succ_j^\prime w$ and $w \succ_{3-j}^\prime x$ for some $j \in \{1,2\}$. We consider the two possibilities: either $j \neq i$ or $j = i$.
		
		\smallskip\noindent
		\textbf{Subcase (i.2.a): $j \neq i$.} Then $w \succ_i^\prime x$. Because $x = \max(T, \succ_i)$, the transitivity of $\succ_i^\prime$ implies $w \succ_{i}^\prime z$ for all $z \in T$. Given the lexicographic order for adding alternatives, we have $z \succ_{3-i}^\prime w$ for all $z \in T$ because $w$ is an extreme. In particular, $w \succ_{i}^\prime x \succ_{i}^\prime y$ and $y \succ_{3-i}^\prime x \succ_{3-i}^\prime w$, hence $c(\{y,w\})= \{y,w\}$. By definition of $\succ^\prime$, $w = \max(T \cup \{w\}, \succ_i^\prime)$ and $y = \max(T \cup \{w\}, \succ_{3-i}^\prime)$; hence $c_{\succ^\prime}(T \cup \{w\}) \subseteq \{y,w\}$. 
		
		Starting from $c(\{y,w\}) = \{y,w\}$, we compute $c(T \cup \{w\})$ directly. We use the observation that, by the induction hypothesis and the definition of $\succ$, for every $z \in T \setminus \{y\}$---including $x$ itself---we have either $y\langle z \rangle w$ or $y \gg z$. 
		
		If \emph{$c(\{y,z\}) = \{y\}$ for some $z \in T$}, the singleton choice implies a higher Borda count for $y$. Hence, we have $c_{\succ^\prime}(T \cup \{w\}) = \{y\}$. Partition $T$ into two sets $T_1$ and $T_2$, where $T_1$ collects the alternatives $z$ with $y\langle z \rangle w$ and $T_2$ collects the alternatives $z$ with $y \gg z$. By repeated application of EXP.ii (start from $\{y,w\}$ towards $T_1 \cup \{y,w\}$), we get $c(T_1 \cup \{y,w\})= \{y,w\}$. Recall that $y \gg z$ and $c(\{w,z\})= \{w,z\}$ for all $z \in T_2$; by repeated application of ATT (start from $T_1 \cup \{y,w\}$ towards $T \cup \{w\}$), we conclude that $c(T \cup \{w\})= \{y\}$. 
		
		Now, suppose instead that \emph{$c(\{y,z\}) = \{y,z\}$ for all $z \in T$}, and so $y$ and $w$ have the same Borda count. Then $b_i(y,T \cup \{w\}; \succ^\prime) = b_{3-i}(w,T \cup \{w\}; \succ^\prime)$, and we have $c_{\succ^\prime}(T\cup \{w\})= \{y,w\}$. Because $c(\{y,w\}) = \{y,w\}$ and $y\langle z \rangle w$ for all $z \in T \setminus \{y\}$, by repeated application of EXP.ii (start from $\{y,w\}$ towards $T \cup \{y,w\}$), we get $c(T \cup \{w\}) = \{y,w\}$.
		
		\smallskip\noindent
		\textbf{Subcase (i.2.b): $j = i$,} that is, $x \succ_i^\prime w$ and $w \succ_{3-i}^\prime x$. There are two possibilities, depending on whether $w$ displaces $y$ from the top position in $\succ_{3-i}$.
		
		If \emph{$w \succ_{3-i}^\prime y$}, then by transitivity $w \succ_{3-i}^\prime z$ for all $z \in T$. Then, because the lexicographic order implies that $w$ is an extreme, we have $z \succ_i^\prime w$ for all $z \in T$. Recall that Case (i) has $c(T)= \{x\}$: therefore, $b_{3-i}(x,T; \succ) > b_{i}(y,T; \succ)$ or else $c(T)$ would contain both $x$ and $y$. This says that $x$ cannot have the lowest rank in $\succ_{3-i}$: there exists some $z \in T$ such that $c(\{x,z\}) = \{x\}$, i.e.\ $x \gg z$. Given that $w$ has the lowest rank in $\succ_i$, we have $b_{3-i}(x,T \cup \{w\}; \succ^\prime) > b_{i}(w,T \cup \{w\}; \succ^\prime)$, and thus $c_{\succ^\prime}(T \cup \{w\}) = \{x\}$. Using similar arguments as in the first branch of Subcase (i.2.a), where $y$ had a higher Borda count, repeated applications of EXP.ii and ATT yield $c(T \cup \{w\})= \{x\}$.
		
		If instead \emph{$y \succ_{3-i}^\prime w$}, then $x$ and $y$ remain top-ranked after the inclusion of $w$: $\max(T \cup \{w\},\succ_i^\prime) = x$ and $\max(T \cup \{w\},\succ_{3-i}^\prime) = y$. Suppose $y \in C$. By definition of $\succ^\prime$, $w \succ_{i}^\prime y$ and $x \langle w \rangle y$. Because $b_{3-i}(x,T; \succ) > b_{i}(y,T; \succ)$, we have $b_{3-i}(x,T \cup \{w\}; \succ^\prime) > b_{i}(y,T\cup \{w\}; \succ^\prime)$ and thus $c_{\succ^\prime}(T \cup \{w\}) = \{x\}$. Given that $x$ and $y$ are top-ranked, each $z \in T \setminus \{x,y\}$ satisfies $x \gg z$ or $y \gg z$ or $x \langle {z} \rangle y$.  We now compute $c(T \cup \{w\})$ directly.
		
		Partition $T \setminus \{x,y\}$ into four sets $T_1, T_2, T_3$ and $T_4$, where $T_1$ collects the alternatives $z$ in $T$ with $x\langle z \rangle y$, $T_2$ collects $z$ such that $x \gg z$ and $y \gg z$, $T_3$ collects $z$ such that $x \gg z$ and $y \| z$, and $T_4$ collects $z$ such that $x \| z$ and $y \gg z$. Start from $c(\{x,y\})=\{x,y\}$. By repeated application of EXP.ii, we get $c(\{x,y\} \cup T_1 \strut)= \{x,y\}$. From there, by repeated application of ATT, we get $c(\{x,y\} \cup T_1 \cup T_2 \strut) = \{x,y\}$. 
		
		Because $b_{3-i}(x,T; \succ) > b_{i}(y,T; \succ)$, we have $|T_3| > |T_4|$ and therefore there exists a bijection $\sigma$ from $T_4$ to a strict subset $T^\prime_3$ of $T_3$. By repeated application of EXP.i (to $z \in T_4$ and $\sigma (z) \in T^\prime_3$), we have $c(T \setminus (T_3 \setminus T^\prime_3)) = \{x,y\}$. Using EXP.ii again, we obtain $c\left( \left[T \setminus (T_3 \setminus T^\prime_3) \right] \cup \{w\} \strut \right) = \{x,y\}$. Finally, by repeated application of ATT on $T_3 \setminus \sigma (T_4)$, we conclude that $y \notin c(T \cup \{w\})$. Therefore, $c(T\cup\{w\})= \{x\}$.
		
		Suppose now $y \in B$. By the induction hypothesis $b_{3-i}(x,T; \succ) > b_{i}(y, T;\ \succ)$. We consider two cases. If $b_{3-i}(x,T;\succ) = b_{i}(y,T;\succ)+1$, then $b_{3-i}(x,T \cup \{w\}; \succ^\prime) = b_{i}(y,T\cup \{w\} ; \succ^\prime)$, which implies $c_{\succ^\prime}(T \cup \{w\}) = \{x,y\}$. Moreover, $b_{3-i}(x,T; \succ) = b_{i}(y,T; \succ)+1$ implies that there exists $z \in T \setminus \{x,y\}$ such that $c(\{x,z\}) = \{x\}$ and $c(\{y,z\}) = \{y,z\}$. By our induction hypothesis, we have $c(T\setminus \{z\})=\{x,y\}$; by EXP.i we get $c(\left[T \setminus \{z\}\right] \cup \{w\}) = \{x,y\}$.
		
		If \emph{$b_{3-i}(x,T;\succ) > b_{i}(y,T;\succ)+1$}, then $b_{3-i}(x,T \cup \{w\}; \succ^\prime) > b_{i}(y,T\cup \{w\}; \succ^\prime)$, which implies $c_{\succ^\prime}(T \cup \{w\}) = \{x\}$. An analogous argument from the above case $y \in C$ shows that $c(T \cup \{w\}) = \{x\}$. 
		
		\smallskip\noindent
		\textbf{Case (ii): $|c(T)| = 2$.} Let $c(T) = \{x,y\}$ and, without loss of generality, assume $x = \max(T, \succ_1)$ and $y = \max(T, \succ_2)$. Given that $x,y$ are both chosen, $b_2 (x, T; \succ) = b_1 (y,T; \succ)$.
		
		\smallskip\noindent
		\textbf{Subcase (ii.a)}: $x,y \in B$. Here, ATT implies $c(T \cup \{w\}) = \{x,y\}$; and, by construction of $\succ^\prime$, $\{x,y\} = c_{\succ^\prime} (T \cup \{w\})$.
		
		\smallskip\noindent
		\textbf{Subcase (ii.b)}: $x \in B, y \in C$. Here, ATT implies $c (T \cup \{w\}) = \{x\}$ and, by construction, $\{ x \} = c_{\succ^\prime}(T \cup \{w\})$. A symmetric argument holds if $x \in C$ and $y \in B$.
		
		\smallskip\noindent
		\textbf{Subcase (ii.c)}: $x,y \in C$. If $x \langle w \rangle y$ holds, then EXP.ii implies $c(T \cup \{w\}) = \{x,y\}$ and, by construction, $c_{\succ^\prime}(T \cup \{w\}) = \{x,y\}$. If instead $x \langle w \rangle y$ does not hold, then either $w \succ_1^\prime x$ or $w \succ_2^\prime y$ (but not both). Assume without loss of generality that $w$ displaces $y$ from its top position, so that $x \succ_1^\prime w$ and $w \succ_2^\prime y$. 
		
		This is the same configuration of subcase (i.2.b): by transitivity, $w = \max(T \cup \{w\}, \succ_2^\prime)$ and $x = \max(T \cup \{w\}, \succ_1^\prime)$. Moreover, $z \succ_1^\prime w$ for every $z \in T$, or else $w \gg z$ would contradict the lexicographic order, which implies $b_1(w, T \cup \{w\}; \succ^\prime) = 0$. On the other hand, $b_2(x, T \cup \{w\}; \succ^\prime) = b_2(x,T; \succ^\prime)$ because $w \succ_2^\prime x$. The argument of subcase (i.2.b) applies, with one difference: there we had $c(T) = \{x\}$, and hence $b_2(x,T; \succ) > b_1(y,T; \succ) \geq 0$, so the comparison was strict; here $c(T) = \{x,y\}$ gives $b_2(x,T; \succ) = b_1(y,T; \succ)$, which is a tie. Two possibilities arise.
		
		If $b_2(x,T; \succ) > 0$, then $b_2(x,T \cup \{w\}; \succ^\prime) > b_1(w,T \cup \{w\}; \succ^\prime)$ and $c_{\succ^\prime}(T \cup \{w\}) = \{x\}$. We show that $c(T \cup \{w\}) = \{x\}$. Note that every $z \in T \setminus \{x\}$ satisfies either $x \gg z$ or $x \langle z \rangle w$. Let $T_1 = \{z \in T \setminus \{x\} : x \langle z \rangle w\}$ and $T_2 = \{z \in T \setminus \{x\} : x \gg z\}$. The set $T_2$ is non-empty because $b_2(x,T) > 0$. Since $c(\{x,w\}) = \{x,w\}$, repeated application of EXP.ii gives $c(\{x,w\} \cup T_1) = \{x,w\} $ and repeated application of ATT (start from $T_1 \cup \{x,w\}$ towards $T \cup \{w\}$) gives $c(T \cup \{w\}) = \{x\}$.
		
		If $b_2(x,T; \succ) = 0$, then $b_2(x,T \cup \{w\}; \succ^\prime) = b_1(w,T \cup \{w\}; \succ^\prime) = 0$ and thus $c_{\succ^\prime}(T \cup \{w\}) = \{x,w\}$. We show that $c(T \cup \{w\}) = \{x,w\}$. There exists no $z \in T$ such that $x \gg z$, so we have $x \langle z \rangle w$ for every $z \in T \setminus \{x\}$. Thus, starting from $c(\{x,w\}) = \{x,w\}$, repeated application of EXP.ii yields $c(T \cup \{w\}) = \{x,w\}$.
	\end{proof}
	
	\subsection{\sc \textbf{Independence of the axioms}}\label{app:indep}
	
	This section collects examples that show the independence of the axioms used in Theorems~\ref{thm2} and~\ref{thm3}. For brevity, we may write $xy$ for the menu $\{x,y\}$, and so on.
	
	\subsubsection*{\sc \textbf{Simple ownership}}
	
	\emph{(PC without AP)} Let $X=\{a,b,x,y\}$. Assume that all binary choices match the order $a\succ b\succ x\succ y$; then PC holds. Let $c( axy )=x$ and $c(bxy)=y$; and, for any other non-binary menu, let $c(A)$ be the $\succ$-maximum. Because $c(ax) = a$ and $c(axy) = x$, $x \vartriangleright y$; similarly, because $c(by) = b$ and $c(bxy) = y$, $y \vartriangleright x$. Then AP is violated.
	
	\medskip
	\noindent
	\emph{(AP without PC)} Let $X=\{x,y,z\}$. Let $c(xy)=x$, $c(yz)=y$, $c(xz)=z$ and $c(xyz)=x$. The only 3-element menu is $\{x,y,z\}$; along with $c(xz) = z$, $c(xyz) = x$ implies $x \vartriangleright y$ and $x \vartriangleright z$. These are the only two prevalence relations: AP holds. Binary choices cycle: PC is violated. 
	
	\subsubsection*{\sc \textbf{Borda ownership}}
	
	Recall that a pair $\{x,y\}$ is a \emph{tie} if $c(\{x,y\})=\{x,y\}$, and a menu is \emph{all ties} if the choice from every pair of its elements is a tie. Moreover, $x \gg y$ if $c(xy)=\{x\}$ and $x \langle y \rangle z$ if $\{x,y,z\}$ is all ties and $c(xyz)=\{x,z\}$. Note that an intermediate exists only if $\{x,y,z\}$ is all ties.
	
	In each example below, all axioms hold except one: we label each example by the violated axiom. 
	
	\medskip\noindent
	\textbf{(PC)} Assume $X=\{x,y,z\}$ with $c(xy)=\{y\}$, $c(yz)=\{z\}$, $c(xz)=\{x,z\}$, and $c(xyz)=\{z\}$. Because $z \gg y \gg x$ but $z \not\gg x$, PC is violated. ATT holds, because the only binding instance is the addition of $y$ to $A=\{ x, z\}$, which matches $c(xyz)=\{z\}$. EFF holds, because $z$ is undominated. Because $\{x,y,z\}$ is not all ties, there are no intermediates; hence, EI, NBI, SI and EXP.ii are vacuously true. Similarly, because there are no four distinct alternatives, EXP.i is vacuously true.
	
	\medskip\noindent
	\textbf{(ATT)} Let $X=\{x,y,z\}$ with $c(xy)=\{x\}$, $c(yz)=\{y,z\}$, $c(xz)=\{x,z\}$, and $c(xyz)=\{x,z\}$. Because $x \gg y$ and $z \not\gg y$, ATT requires that the choice after adding $y$ to the menu $xz$ is only $x$, yet $z$ survives: ATT is violated. All other axioms hold or are vacuously true exactly as in the previous example.
	
	\medskip\noindent
	\textbf{(EFF)} Let $X=\{x,y,z,w\}$ with $c(xy)=\{x\}$, $c(zw)=\{z\}$; all other pairs are ties; $c(xyz)=c(xyw)=\{x\}$, $c(xzw)=c(yzw)=\{z\}$; and $c(X)=\{x,y,z\}$. Because $x \gg y$ but $y \in c(X)$, EFF is violated. PC holds trivially. The four binding instances for ATT are $\{xz\} \cup \{y\}$, $\{xw\} \cup \{y\}$, $\{xz\} \cup \{w\}$ and $\{yz\} \cup \{w\}$: they all prescribe choices that match $c ( \cdot)$. No triplet is all ties, so EI, NBI, SI and EXP.ii are vacuously true. The only binding instance for EXP.i is $A=\{x,z\}$ with $y,w \notin A$, that prescribes $x,z \in c(xywz)$. 
	
	\medskip\noindent
	\textbf{(EI)} Let $X=\{x,y,z\}$ with all pairs tied and $c(xyz)=\{x,y,z\}$. The choice for the all ties triplet is not a doubleton, so EI is violated. Moreover, there exists no intermediate so NBI, SI and EXP.ii are vacuously true. Finally, there is no dominance, so PC, ATT, EFF and EXP.i are also vacuously true.
	
	\medskip\noindent
	\textbf{(NBI)} Let $X=\{x,y,z,w\}$ be all ties, with intermediates over triplets described by $x \gtrdot y \gtrdot w \gtrdot z$ (namely, $c(xyz)=\{x,z\}$, $c(xyw)=\{x,w\}$, $c(ywz)=\{y,z\}$, and $c(xwz)=\{x,z\}$) and $c(X)=\{x,y,z\}$. Because $x\langle y \rangle z$ but $y \in c(X)\supseteq \{x,z\}$, NBI fails. There is no dominance relation so PC, ATT, EFF and EXP.i are vacuously true. EI is trivially satisfied. SI holds because the intermediates are described by an order. EXP.ii requires $x,z \in c(X)$ (via the triplets $xyz$ or $xwz$); $x,w \in c(xyw)$, and $y,z \in 	c(ywz)$, all of which hold. 
	
	\medskip\noindent
	\textbf{(SI)} Let $X=\{x,y,z,w\}$ with all pairs tied except $c(xw)=\{x\}$ and $c(wz)=\{z\}$; $c(xyz)=\{x,z\}$, $c(xyw)=\{x\}$, $c(ywz)=\{z\}$, $c(xwz)=\{x,z\}$; and $c(X)=\{x,z\}$. The only intermediate is $x\langle y \rangle z$; because $c(yw) = \{y,w\}$, SI requires either $x\langle y \rangle w$ or $w \langle y \rangle z$; but neither is possible because $\{x,w\}$ and $\{w,z\}$ are not ties. Hence, SI is violated. The only dominance relations are $x \gg w$ and $z \gg w$, so PC holds trivially. The only binding instances for ATT arise when adding $w$ to a tied menu, and choices match what is required. EI holds because the only all ties triplet is $\{x,y,z\}$, and $x \langle y \rangle z$. NBI requires $y \notin c(A)$ for $A \supseteq\{x,z\}$, which holds. 
	EXP.i has no binding instance. The only binding instances for EXP.ii require $\{xz\}\subseteq c(xyz) \cap c(X)$, which holds.
	
	\medskip\noindent
	\textbf{(EXP.i)} Let $X=\{x,y,w,z\}$ with all pairs tied except $c(xz)=\{x\}$ and, $c(yw)=\{y\}$; $c(xyz)=\{x\}$, $c(xyw)=\{y\}$, $c(xwz)=\{x\}$, $c(ywz)=\{y\}$; and $c(X)=\{x\}$. Taking $A=\{ x, y \}$ and the pair $\{w,z\}$, the premise of EXP.i is satisfied, but $y \notin c(X)$: EXP.i is violated. The only dominance relations are $x \gg z$ and $y \gg w$, so PC holds trivially. There are no binding instances for ATT and EFF, so they are vacuously true. Similarly, because there are no intermediates, EI, NBI, SI and EXP.ii are vacuously true.
	
	\medskip\noindent
	\textbf{(EXP.ii)} Let $X=\{x,y,z,w\}$ with all pairs tied. Fix the order $x\gtrdot y\gtrdot z\gtrdot w$ and let choice over each triplet give an intermediate that is the median according to $\gtrdot$; e.g., $c(\{x,z,w\})=\{x,w\}$, so $x\langle z\rangle w$. Finally, set $c(X)=\{y,z\}$. For $A=\{x,z,w\}$, we have $c(A)=\{x,w\}$; because $y \notin c(A)$ and $x \langle y \rangle w$ holds, EXP.ii requires $x,w \in c(A\cup\{y\})=c(X)$, which is violated by $c(X)=\{y,z\}$. Since every pair is tied, $\gg$ never holds so PC, ATT, EFF, and EXP.i are vacuously true. EI holds by construction. SI holds because each intermediate reappears exactly as predicted (there are four instances to check). NBI holds because no triplet has $\{y,z\}$ as a pair of extremes, so nothing conflicts with $c(X)=\{y,z\}$.
	
	\subsection{\sc \textbf{Proofs for Section~\ref{sec:ident}}}\label{app:ident}
	
	\subsubsection*{\sc \textbf{Proposition~\ref{prop:id1}}}
	
	\begin{proof} We start with the proof of (1) and then move to the proof of (2).
		
		\smallskip\noindent
		\emph{(Necessity)} Let $(\succ_1',\succ_2',\omega)$ represent $c$ and let $(A,B)$ be a reversal. If Agent~$i$ with preference $\succ_i$ owns both menus, then $c(B)=\max(B,\succ_i)$ is unbeaten in $B$ and hence in $A \subset B$. Then $c(B) \in A$ is the unique maximal element of $A$, and we have the contradiction $c(A)=c(B)$.
		
		\smallskip\noindent
		\emph{(Sufficiency)} Let $\omega$ be monotone and assign different owners to the two menus of every reversal. Let $\Sigma_i$ denote the collection of menus owned by Agent~$i$. Without loss of generality, let $\omega(X)=2$. Then $\Sigma_2$ is an upper set and $\Sigma_1$ a downset. For $i=1,2$ let $x \succ_i y$ if $x=c(A)$ and $y \in A \setminus \{x\}$ for some $A \in \Sigma_i$. We show that $\succ_i$ is asymmetric. Suppose by contradiction that $x=c(A)$ with $y \in A$ and $y=c(A')$ with $x \in A'$, where $A,A' \in \Sigma_i$. If $\Sigma_i$ is a downset, then $B=A \cap A'$ contains $x,y$ and belongs to $\Sigma_i$. As $x \in B \subseteq A$ and $(B,A)$ is not a reversal, $c(B)=x$; likewise $c(B)=y$, which is a contradiction. If $\Sigma_i$ is an upper set, let $z = c(A \cup A')$ and say $z \in A$. Since $(A,A \cup A')$ is not a reversal, $c(A)=z$, so $z=x$. But $x \in A'$, so the same argument gives $c(A')=x \neq y$, which is a contradiction. The case $z \in A'$ is symmetric.
		
		\smallskip\noindent
		\emph{(2)} In a representation, $c(A)=\max(A,\succ_i')$ is the unique maximal element for every menu $A$ owned by Agent~$i$, so no $y \in A$ satisfies $y \succ_i' c(A)$. For $A=\{x,y\}$ this implies $c(A)=x$ if and only if and only if $x \succ_i' y$.
	\end{proof}
	
	\subsubsection*{\sc \textbf{Proposition~\ref{prop:id2}}}
	
	\begin{proof} 
		Assuming that Agent~1 owns the pairs, $\succ_1'=\succ_1$. Agent~1 cannot own a menu $A \in \mathcal{D}$, so $c(A) = \max(A,\succ_2')$; chaining these relations gives~(2). Conversely, assume (1)--(3). Pairs and menus outside $\mathcal{D}$ owned by Agent~1 are chosen by $\succ_1$. For $A \in \mathcal{D}$, some $z \in A$ beats $c(A)$ in the pairs, so $c(A) \vartriangleright y$ for every $y \in A \setminus \{c(A)\}$; thus, (2) gives $c(A) = \max(A,\succ_2')$. Menus given to Agent~2 under (3) are fine by assumption.
	\end{proof}
	
	\subsubsection*{\sc \textbf{Proposition~\ref{prop:id3}}}
	
	We start with an observation. Let $A$ be a menu, $B$ the first block that meets $A$, and $A_0 = A \cap B$. Every alternative in $A_0$ dominates every alternative in $A \setminus A_0$, so it ranks higher in both orders: the top alternatives of $A$ lie in $A_0$. Moreover, for $x \in A_0$ the two Borda counts in $A$ equal those in $A_0$ plus $|A \setminus A_0|$. Because the determinants of ownership are unchanged, we have $c(A) = c(A_0)$. This holds for any pair of orders which generate the same blocks.
	
	\begin{proof}
		\emph{(Sufficiency)}
		The pair $\succ^J$ is a pair of orders, with the same $\gg$ (dominance) and $\|$ (incomparability) relations as $\succ$. Hence, they generate the same blocks. Given a menu $A$, our preliminary observation implies that it suffices to compare $\succ$ and $\succ^J$ on $A_0 \subseteq B$. There $\succ^J$ is either $\succ$ or $(\succ_2,\succ_1)$, which generate the same $c$.
		
		\smallskip\noindent
		\emph{(Necessity)} Let $\succ'$ also generate $c$. It has the same $\gg$ and the same blocks. If $x \gg y$, both orders of both representations rank $x$ above $y$; this covers all pairs in different blocks. It remains to check incomparable pairs inside a block $B$ with at least two alternatives. For an incomparable pair $\{a,b\}$, let $s(a,b)$ be the agent whose order in $\succ$ ranks $a$ above $b$; define $s^\prime$ similarly for $\succ^\prime$. We show that either $s'=s$ on all incomparable pairs of $B$, or $s'=3-s$ on all of them. The argument requires the following two ancillary facts.
		
		\noindent
		\emph{(i) Incomparability connects $B$.} Consider the graph whose vertices are the alternatives of a block $B$ and whose edges join incomparable pairs. Let $P$ be a connected component of this graph, and suppose $Q = B \setminus P$ is non-empty. Then no element of $P$ is incomparable to an element of $Q$. Fix $q \in Q$. If $p_0 \gg q$ for some $p_0 \in P$, then $p \gg q$ for all $p \in P$: if not, along a chain $p_0 = z_0,\dots,z_t = p$ of incomparable pairs, $z_l \gg q$ and $q \gg z_{l+1}$ would imply $z_l \gg z_{l+1}$ by transitivity, contradicting $z_l \| z_{l+1}$. Similarly, if $q \gg p_0$ for some $p_0$, then $q \gg p$ for all $p \in P$. Let $Q^-$ be the elements of $Q$ that dominate all of $P$, and $Q^+$ those dominated by all of $P$. (Denote an element of $Q^-$ by $q^-$ and an element of $Q^+$ by $q^+$.) If $Q^-$ is empty, $P$ dominates the rest of $B$ and all later blocks, so it is dominant in $X \setminus (B_1 \cup \dots \cup B_{k-1})$ and strictly smaller than $B=B_k$, contradicting the definition of $B_k$. If $Q^-$ is non-empty, then $q^- \gg p \gg q^+$ gives $q^- \gg q^+$, so $Q^-$ dominates $P$, $Q^+$ and all later blocks. It is dominant and strictly smaller than $B$, which is again a contradiction.
		
		\noindent
		\emph{(ii) Adjacent incomparable pairs are linked by $c$.} Let $\{x,y\}$ and $\{y,z\}$ be incomparable pairs with $x \neq z$. Let $t=1$ if the agent whose order ranks $x$ above $y$ also ranks $y$ above $z$, and $t=0$ otherwise. We claim that $t$ is determined by $c$. Suppose first that $x \| z$. Then on $\{x,y,z\}$ the two orders are reverses of each other, so one alternative $m$ is in the middle for both, and the choice from the triplet is the doubleton of the other two; $c$ reveals $m$. If $m=y$ then $t=1$, and if $m \in \{x,z\}$ then $t=0$. Suppose next that $x$ and $z$ are comparable, so both orders agree on them. If $t=1$ for some agent $i$, then $x \succ_i y \succ_i z$, while the other agent reverses both pairs and ranks $z \succ_{3-i} y \succ_{3-i} x$; the orders then disagree on $x,z$, a contradiction. So $t=0$, which is again read off from $\gg$.
		
		We return to the necessity proof. Take adjacent incomparable pairs $\{x,y\}$ and $\{y,z\}$ in $B$. Under $\succ$, $s(y,z)=s(x,y)$ if $t=1$ and $s(y,z)=3-s(x,y)$ if $t=0$; the same formula holds for $s^\prime$ with the same $t$. Hence $s^\prime (y,z)=s(y,z)$ if and only if $s^\prime (x,y)=s(x,y)$. By (i), any two incomparable pairs of $B$ are joined by a chain of adjacent ones, so either $s^\prime=s$ everywhere in $B$ or $s^\prime=3-s$ everywhere. Let $J$ be the set of blocks of the second kind. For an incomparable pair $\{a,b\}$ in a block of $J$, $\succ_i^\prime$ ranks $a$ above $b$ if and only if $s^\prime(a,b)=i$, if and only if $s(a,b)=3-i$, if and only if $\succ_{3-i}$ ranks $a$ above $b$. Thus $\succ^\prime = \succ^J$.
		
		Concerning the number of representations, any block with at least two alternatives contains an incomparable pair, so its two orders differ and exchanging them changes the pair. Hence the $\succ^J$, with $J$ ranging over subsets of these $m$ blocks, are all distinct. Moreover $(\succ^J_2,\succ^J_1)=\succ^{J^c}$, where $J^c$ is the complement of $J$ among these blocks, so relabeling pairs the representations off and leaves $2^{m-1}$ cases for $m \ge 1$. If $m=0$, all blocks are singletons and $\succ_1=\succ_2$, so the representation is unique.
	\end{proof}
	
	\subsection{\sc \textbf{An axiomatization of TAT}}\label{app:TAT}
	
	\citet{eliaz2011choosing} provide a behavioral characterization for the top-and-top (TAT) procedure, based on four independent axioms. Their first axiom (A0') states $|c(A)| \le 2$ for any menu $A$. A cardinality assumption seems a blunt way to ensure dual ownership. We drop A0' and offer an alternative axiomatization, adapting an argument in \citet{manzini2018dual} to our deterministic setup. Recall the other three axioms in \citet{eliaz2011choosing}. (For ease of comparison, we keep the original names.)
	
	\medskip
	\noindent
	\textbf{A1}: if $a \in c(A)$ and $a \in B \subset A$, then $a \in c(B)$.
	\medskip
	
	\noindent
	\textbf{A4}: if $\{a\} = c(A)$ and $a \in B \subset A$, then $\{a\} = c(B)$.
	\medskip
	
	\noindent
	\textbf{A5}: if $a \in c(A_1) \cap c(A_2) \cap c (A_1 \cap A_2)$ and if $|c (A_1 \cap A_2)|=2$, then $a \in c (A_1 \cup A_2)$.
	\medskip
	
	We say that an alternative $b \notin A$ is \emph{neutral} for another alternative $a \in A$ with respect to a menu $A$ if either (i) $\{a\} = c(A)$ and $\{a\} = c(A \cup \{b\})$, or (ii) $|c(A)|=2$, $a \in c(A)$ and $a \in c(A \cup \{b\})$, or (iii) $a \notin c(A)$ and $a \notin c(A \cup \{b\})$; see \citet[p.~1158]{manzini2014stochastic}. We consider the following consistency axiom (C):
	
	\medskip
	\noindent
	\textbf{C}: if $b$ is neutral for every $a$ in $A$, then $b \notin c(A \cup \{b\})$.
	\medskip
	
	Our first intermediate result shows that the combination of A1 and C is a sufficient condition for A0'.
	
	\begin{proposition}\label{prop:A0}
		A1 and C imply $|c(A)| \le 2$ for any menu $A$.
	\end{proposition}
	\begin{proof}
		Suppose by contradiction that $b_1,b_2,b_3 \in c(A)$. By A1, $c(\{b_1, b_2 \}) = \{b_1, b_2 \}$ which has size 2. Moreover, $b_1, b _2 \in c(\{b_1,b_2,b_3\})$ by A1, and so $b_3$ is neutral for $b_1$ and $b_2$ with respect to $\{b_1,b_2\}$. Then C implies $b_3 \notin c(\{b_1,b_2,b_3\})$, which contradicts A1. 
	\end{proof}
	
	Replacing A0' with C in the axiomatization by \citet{eliaz2011choosing}, we obtain a new characterization of the TAT procedure. 
	
	\begin{proposition}\label{prop:TAT}
		A choice correspondence $c$ satisfies A1, A4, A5 and C if and only if it is a TAT procedure.
	\end{proposition}
	\begin{proof}
		Sufficiency is a direct consequence of Proposition~\ref{prop:A0} and \citet[Proposition, p.~217]{eliaz2011choosing}. We prove necessity for $C$. Suppose by contradiction that $b$ is neutral for every $a$ in some $A$, but $b \in c (A \cup \{b\})$. Because $c$ is a TAT procedure, there are two cases: if $|c(A)| = 1$, then $c(A) = \{a\}$ for some $a \in A$ and by neutrality $\{ a \} = c(A \cup \{b\})$, so $b \notin c (A \cup \{b\})$; if $|c(A)| = 2$, then $c(A) = \{a_1,a_2\}$ for some $a_1,a_2 \in A$ and by neutrality $a_i \in c(A \cup \{b\})$ for $i=1,2$, which saturates $c(A)$ and so $b \notin c (A \cup \{b\})$.
	\end{proof}
	
	The four axioms in Proposition~\ref{prop:TAT} are independent. This can be shown using the same counterexamples as in \citet[pp.~218--219]{eliaz2011choosing}, after replacing A0' with C. In particular, for their counterexample~(i) assume that $X$ has three elements.
	
\end{document}